\documentclass[10pt,twocolumn]{revtex4}
\usepackage{graphicx}
\usepackage{amsmath,latexsym,amssymb,amsthm,array,bm}
\usepackage{tikz}
\usepackage{bm}
\usepackage{multirow}

\theoremstyle{plain}
\newtheorem{theorem}{Theorem}
\newtheorem{lemma}[theorem]{Lemma}
\newtheorem{proposition}[theorem]{Proposition}

\theoremstyle{definition}
\newtheorem{definition}[theorem]{Definition}

\newtheorem{corollary}[theorem]{Corollary}

\newtheorem{remark}[theorem]{Remark}

\theoremstyle{remark}

\newlength{\blank}
\newcommand{\cpy}[1]{\tilde{#1}} 
\def\A{\mathrm{A}}
\def\At{\cpy{\A}}
\def\B{\mathrm{B}}
\def\C{\mathrm{C}}
\def\Q{\mathrm{Q}}
\def\R{\mathrm{R}}
\def\Qt{\cpy{\Q}}

\def\goc{\bm{\Omega}} 
\def\mclass{\goc}

\def\subops{\bm{\mathrm{subops}}} 
\def\lm{\mathcal{L}} 
\def\ops{\bm{\mathrm{ops}}} 
\def\states{\bm{\mathrm{states}}} 
\def\mops{\bm{\mathrm{mops}}} 
\def\tests{\bm{\mathrm{T}}} 
\def\gmap{\mathcal{N}}
\def\transmap{\mathrm{t}}

\newcommand{\nc}{\newcommand}
\def\mc{\mathcal}

\def\ox{\otimes}
\def\Ox{\bigotimes}

\nc{\bra}[1]{\langle#1|}
\nc{\ket}[1]{|#1\rangle}
\nc{\ketbra}[2]{|#1\rangle\!\langle#2|}
\nc{\braket}[2]{\langle#1|#2\rangle}
\nc{\innerproduct}[2]{\langle #1, #2 \rangle}
\nc{\eref}[1]{(\ref{#1})}

\nc{\pure}[1]{\ketbra{#1}{#1}}
\nc{\partrans}[2]{\Gamma_{#1}\left(#2\right)}

\def\squareforqed{\hbox{\rlap{$\sqcap$}$\sqcup$}}
\def\qed{\ifmmode\squareforqed\else{\unskip\nobreak\hfil

\penalty50\hskip1em\null\nobreak\hfil\squareforqed
\parfillskip=0pt\finalhyphendemerits=0\endgraf}\fi}
\def\endenv{\ifmmode\;\else{\unskip\nobreak\hfil
\penalty50\hskip1em\null\nobreak\hfil\;
\parfillskip=0pt\finalhyphendemerits=0\endgraf}\fi}

\nc\supp{\mathrm{supp}}

\nc{\mathword}[1]{\mathinner{\mathrm{#1}}}

\def\tr{\mathword{Tr}}

\def\hc{^{\dagger}} 

\def\Hs{\mc{H}}

\DeclareRobustCommand\idop{\leavevmode\hbox{\small1\normalsize\kern-.33em1}}
\DeclareRobustCommand\scriptidop{
\leavevmode\hbox{\fontsize{7}{8}\selectfont 1\scriptsize\kern-.33em1}}

\def\1{\idop}

\def\id{\mathrm{id}}

\def\chan{\mathcal{E}}
\def\src{\mathcal{S}}
\def\code{\mathcal{Z}}

\def\chS{\chan_{\outS\qcon\A}}
\def\iS{\rho_{\A}}
\def\icS{\rho_{\At}}
\def\oS{\sigma_{\outS}}
\def\meS{\Phi_{\At\A}}
\def\bS{\At\outS}

\def\purho{\rho_{\At\A}}

\global\long\def\tr{\mathrm{Tr}}

\def\id{\mathrm{id}}

\def\eA{\mathrm{A_E}}
\def\eB{\mathrm{B_E}}
\def\sT{\mathrm{G}}

\def\ient{\Psi}
\def\ent{\Phi}
\def\mV{w}
\def\dmV{\hat{w}}
\def\mRV{W}
\def\dmRV{\hat{W}}

\def\outS{\mathrm{B}}
\def\enc{\mathcal{C}}
\def\chan{\mathcal{E}}

\def\id{\mathrm{id}}
\def\code{\mathcal{Z}}
\def\src{\mathcal{S}}

\def\T{\mathrm{T}}

\def\half{\frac{1}{2}}

\def\LO{\bm{\mathrm{L}}}
\def\oneway{\bm{\mathrm{LC1}}}
\def\PPT{\bm{\mathrm{PPT}}}
\def\all{\bm{\mathrm{ALL}}}

\def\mmap{\mathcal{M}_{\A' \B' \qcon \A \B}}
\def\qcon{|}

\def\SE{\rm{\mathrm{E}}}

\newcommand{\rBeta}[4]{\beta^{#1}_{#2}(#3\|#4)} 
\newcommand{\rDH}[4]{D^{#1}_{#2}(#3\|#4)} 

\newcommand{\rBch}[4]{\beta^{#1}_{#2}(#3,#4)} 
\newcommand{\rIHch}[4]{I^{#1}_{#2}(#3,#4)} 
\newcommand{\Mi}[3]{M_{#1}(#2,#3)} 
\newcommand{\M}[2]{M_{#1}(#2)} 
\newcommand{\Mei}[3]{M^{\SE}_{#1}(#2,#3)}
\newcommand{\Me}[2]{M^{\SE}_{#1}(#2)} 

\definecolor{darkred}{rgb}{0.8,0,0}
\definecolor{grey}{rgb}{0.5,0.5,0.5}

\begin{document}
\title{Finite blocklength converse bounds for quantum channels} 
\author{William Matthews}
\affiliation{Statistical Laboratory,
University of Cambridge,
Wilberforce Road, Cambridge,
CB3 0WB
England}
\email{will@northala.net}
\thanks{This paper was presented in part at Quantum Information
Processing 2013}
\author{Stephanie Wehner}
\affiliation{Centre for Quantum Technologies,
National University of Singapore,
3 Science Drive 2,
Singapore, 117543 }

\begin{abstract}
    We derive upper bounds on the rate of transmission
    of classical information over quantum channels by
    block codes with a given blocklength and error probability,
    for both entanglement-assisted and unassisted codes,
    in terms of a unifying framework of quantum hypothesis testing
    with restricted measurements.
    Our bounds do not depend on any special property of the
    channel (such as memorylessness) and
    generalise both a classical converse of
    Polyanskiy, Poor, and Verd\'{u} as well as a quantum converse of
    Renner and Wang, and have a number of desirable properties.
    In particular our bound on entanglement-assisted codes
    is a semidefinite program and for memoryless
    channels its large blocklength limit is the well known
    formula for entanglement-assisted capacity
    due to Bennett, Shor, Smolin and Thapliyal.
\end{abstract}

\maketitle

\section{Introduction}
This work is concerned with the transmission of
classical information over quantum channels by means
of block codes. This is a central subject of study
in quantum information theory, and the \emph{asymptotic}
rates of transmission in the large blocklength limit,
for various types of code and channel,
are the subject of celebrated theorems and intriguing open problems.
A more fundamental problem, of both theoretical and practical
interest, is to obtain upper (or \emph{converse})
and lower (or \emph{achievability}) bounds on
the optimal transmission rate
for a given error probability $\epsilon$ and
\emph{finite} blocklength $n$.

Without assumptions
on the structure of the operation
implemented by $n$ channel uses (e.g. independence),
there is only a notational
difference between coding for $n$ uses of a channel and
coding for one use of a larger composite channel,
so bounds which apply in this setting are also
known as `one-shot' bounds.
These are the subject of a number of recent results in
quantum information~\cite{MD,RennerWang,RR,DH,DMHB} and remain
an active topic of research
in classical information \cite{PPV,2013-Polyanskiy}.
All bounds referred to in the remainder of this introduction
are of this type.

Mosonyi and Datta \cite{MD}, Wang and Renner \cite{RennerWang}
and Renes and Renner \cite{RR}
have given converse and achievability bounds for classical-quantum channels.
In \cite{DH} Datta and Hsieh derive converse and achievability results for
\emph{entanglement-assisted coding} over quantum channels in terms of
smoothed min- and max-entropies.

Polyanskiy, Poor and Verd\'{u}~\cite{PPV} identified a very
general approach to \emph{classical} converse bounds which they call
the `meta-converse'. The bounds they obtain are given
in terms of a classical hypothesis testing problem.
In this paper we further generalise this approach to the
quantum setting, obtaining novel converse bounds for both
entanglement-assisted and unassisted coding over general
quantum channels. Our bounds are given in terms of a
\emph{quantum} hypothesis testing problem on a bipartite system.
In the bounds for unassisted codes the measurements used for the
hypothesis test obey certain locality restrictions with respect to the bipartition.

Section \ref{prelims} gives a brief review of the mathematical
framework that we work with,
introduces the classes of restricted operations we use (\ref{classes})
and quantum hypothesis testing with these restrictions (\ref{sec:hypoTest}),
and defines precisely entanglement-assisted and unassisted codes
(\ref{codes}).
Section \ref{summary} states our converse
bounds and summarises the properties of these bounds, which we establish
in Section \ref{sec:props}. Section \ref{sec:fbc} proves the
converse bounds via a quantum `meta-converse' result.

As an example of the application of our bound for entanglement-assisted
codes, in Section \ref{sec:eg} we show how to compute it exactly for $n$ uses of a depolarising channel. In Section \ref{sec:sconv} we discuss the relationship of this work to existing results
on strong converse bounds for quantum channels, and to security
proofs in the noisy-storage model.
We conclude in Section \ref{sec:conc}, where we mention
some open problems.

\def\bmo{\mathcal{T}}
\def\sdm{\mathbf{d}}

\section{Preliminaries}\label{prelims}
As usual, a quantum system $\Q$ is associated with
a Hilbert space $\mathcal{H}_{\Q}$. By the dimension
of the system $\dim(\Q)$
we mean the dimension of the associated space.
This work deals only with finite-dimensional systems.
The space of linear maps from
$\Hs_{\Q}$ to $\Hs_{\R}$,
we denote by $\lm(\Hs_{\Q},\Hs_{\R})$,
with the abbreviation
$\lm(\Hs_{\Q})$ for the space $\lm(\Hs_{\Q},\Hs_{\Q})$
of linear \emph{operators} on $\Hs_{\Q}$.

By a \emph{state} of $\Q$ we mean a density
(i.e. positive semidefinite, trace one)
operator on $\lm(\mathcal{H}_{\Q})$.
We denote the set of all states of $\Q$ by $\states(\Q)$.
We omit tensor products, when the subscripts make it clear
on which systems the operators act e.g.
$\rho_{\A} \sigma_{\B} = \rho_{\A} \otimes \sigma_{\B}$.

By a \emph{sub-operation} with input system $\mathrm{A}$ and
output system $\mathrm{B}$ we mean a linear map from
$\lm(\mathcal{H}_{\mathrm{A}})$ to
$\lm(\mathcal{H}_{\mathrm{B}})$ which is
completely positive and trace non-increasing.
We denote the set of these by $\subops(\A\to\B)$.
If a sub-operation is trace preserving, then it is
called an \emph{operation}. We denote the subset
of operations in $\subops(\A\to\B)$ by $\ops(\A\to\B)$.

As a completely positive map, any sub-operation
$\mc{N}_{\B|\A} \in \subops(\A\to\B)$ has a
(non-unique) representation
\begin{equation}
    \mc{N}_{\B|\A}: X_{\A} \mapsto \sum_{i} M_{i} X_{\A} M_{i}^{\dagger}
\end{equation}
in terms of linear operators
$M_{i} \in \lm(\Hs_{\A},\Hs_{\B})$ 
called \emph{Kraus operators}. 
The trace non-increasing condition is equivalent to
$\sum_i M_{i}^{\dagger} M_{i} \leq \1_{\A}$, and $\mc{N}_{\B|\A}$
is trace preserving iff $\sum_i M_{i}^{\dagger} M_{i} = \1_{\A}$.

Given a linear map $\mc{N}_{\B|\A}: \lm(\Hs_{\A}) \to \lm(\Hs_{\B})$
its \emph{adjoint} map is defined to be
the unique linear map
$\mc{N}^{\dagger}_{\B|\A}: \lm(\Hs_{\B}) \to \lm(\Hs_{\A})$ such that,
for all $X_{\A} \in \lm(\Hs_{\A})$ and $Y_{\B} \in \lm(\Hs_{\B})$,
we have $\tr_{\B} Y_{\B}^{\dagger} \mc{N}_{\B|\A} X_{\A}
= \tr_{\A} (\mc{N}^{\dagger}_{\A|\B} Y_{\B})^{\dagger}  X_{\A}$.
If $\mc{N}_{\B|\A}$ is an operation, then $\mc{N}^{\dagger}_{\A|\B}$
is completely positive and unital (i.e. it maps the identity operator
to the identity operator).

Processes which produce a classical outcome in some
set $W$, can be represented by an \emph{instrument}.
For our purposes, it suffices to consider finite $W$ and we
can represent an instrument by a collection of
suboperations indexed by $W$,
$\{ \mc{N}^{(w)}_{\B|\A} \}_{w \in W}$, which
sum to an operation:
$\sum_{w \in W} \mc{N}^{(w)}_{\B|\A} \in \ops(\A \to \B)$.
If the instrument is applied to a state $\rho_{\A}$
then the probability of outcome $w$ is
$\tr_{\B} \mc{N}^{(w)}_{\B|\A} \rho_{\A}$,
and state of $\B$ conditioned on outcome $w$ is
$(\mc{N}^{(w)}_{\B|\A} \rho_{\A})/
(\tr_{\B} \mc{N}^{(w)}_{\B|\A} \rho_{\A} )$.
If $M^{(w)}_i$ are Kraus operators for $\mc{N}^{(w)}_{\B|\A}$, then
$\tr_{\B} \mc{N}^{(w)}_{\B|\A} \rho_{\A}
= \sum_{i} \tr_{\B} M^{(w)}_i \rho_{\A} (M^{(w)}_i)^{\dagger}
= \tr_{\A} E(w)_{\A} \rho_{\A}$
where $E(w)_{\A} := \sum_{i} (M^{(w)}_i)^{\dagger} M^{(w)}_i$,
and the equation follows from the linearity and cyclic property of trace.
The \emph{POVM} (positive operator-valued measure)
$\{ E(w)_{\A} \}_{w \in W}$
and the state $\rho_{\A}$ determine the distribution of the outcome.
The fact that $\{ \mc{N}^{(w)}_{\B|\A} \}_{w \in W}$ is an instrument
implies that POVM elements satisfy $E(w)_{\A} \geq 0$ and
$\sum_{w \in W} E(w)_{\A} = \1_{\A}$.

It is convenient to assume that every quantum system $\Q$
comes equipped with a canonical orthonormal basis which we
call the \emph{classical} basis, and whose members we denote by
$\ket{i}_{\Q}$ for $i = 1, \ldots, \dim(\Q)$.
Given two systems $\Q$ and $\Qt$ of the same dimension,
we denote by $\id_{\Qt|\Q}$ the ``identity'' operation
from $\Q$ to $\Qt$ which is defined via its action on
the classical bases of the systems, thus
\begin{equation}
    \id_{\Qt|\Q} : \ketbra{i}{j}_{\Q} \mapsto
    \ketbra{i}{j}_{\Qt}.
\end{equation}
We also define the \emph{transpose map} on a system $\Q$ by its action
on the classical basis
\begin{equation}
    \transmap_{\Q|\Q} : \ketbra{i}{j}_{\Q} \mapsto
    \ketbra{j}{i}_{\Q}.
\end{equation}
It is important to note that $\transmap_{\Q|\Q}$ is positive but
not completely positive, so it is not an operation.
Given an operator $X_{\Q}$ we will also make use of the standard
notation for its transpose
$X_{\Q}^{\T} := \transmap_{\Q|\Q} X_{\Q}$. 
Similarly, the complex conjugate $X^{\ast}_{\Q}$ of
$X_{\Q}$ is defined by taking the classical basis to be real.
Therefore, the adjoint $X_{\Q}^{\dagger}$ of an operator $X_{\Q}$
satisfies $X_{\Q}^{\dagger} = (X_{\Q}^{\T})^{\ast}$.

\subsection{Classes of operation on bipartite systems}\label{classes}

Let $\ops^{\goc}(\A:\B \to \A':\B')$
denote the set of all operations
taking states of the bipartite system $\A : \B$
to states of $\A' : \B'$, which belong to class $\goc$.
We insert the colon to make explicit the relevant
bipartitions of the input and output systems.

We call an operation a \emph{measurement operation}
if, for any input, its output is diagonal in the classical
basis of its output system(s). It is worth emphasising
that the classical basis of a composite system is the
\emph{product basis} formed from the classical bases
of its constituents. We denote the subset of measuring
operations in $\ops^{\goc}(\A:\B \to \A':\B')$ by
$\mops^{\goc}(\A:\B \to \A':\B')$. Measuring operations
are also called ``quantum-classical'', or ``q-c'', operations
in the literature.

An operation
$\gmap_{\A' \B' \qcon \A \B} \in \ops(\A:\B \to \A':\B')$
belongs to the class $\PPT$ if it is positive-partial-transpose preserving,
i.e. if $\transmap_{\B'|\B'} \gmap_{\A' \B' \qcon \A \B}
\transmap_{\B|\B}$ is
completely positive.

$\gmap_{\A' \B' \qcon \A \B}$ belongs to $\oneway$
if it can be implemented by
local operations and one-way classical communication from Alice to Bob:
Alice performs any instrument on her side, generating a classical
outcome $a$ which she sends to Bob. Bob uses $a$ to determine
which operation he applies. Such an operation can be written in the form
\begin{equation}
	\gmap_{\A' \B' \qcon \A \B}
	= \sum_{a}
	\mathcal{F}^{(a)}_{\A' \qcon \A}
	\mathcal{D}^{(a)}_{\B' \qcon \B}
\end{equation}
where, for each $a$ 
$\mathcal{F}_{\A' \qcon \A}^{(a)} \in \subops(\A\to\A')$,
$\sum_{a} \mathcal{F}_{\A' \qcon \A}^{(a)} \in \ops(\A\to\A')$,
and $\mathcal{D}^{(a)}_{\B' \qcon \B} \in \ops(\B\to\B')$.
Throughout, we will omit tensor products between operations
when it is clear which systems they act on from the subscripts.

$\gmap_{\A' \B' \qcon \A \B}$ belongs to $\LO$ if it can be implemented by
local operations and shared randomness, which means it can be written
\begin{equation}
	\gmap_{\A' \B' \qcon \A \B}
    = \sum_{r} p_r
    \mathcal{F}^{(r)}_{\A' \qcon \A}
    \mathcal{D}^{(r)}_{\B' \qcon \B}
\end{equation}
where $\mathcal{F}^{(r)}_{\A' \qcon \A} \in \ops(\A\to\A')$ and 
$\mathcal{D}^{(r)}_{\B' \qcon \B} \in \ops(\B\to\B')$ and $p_r$
is the probability of the shared random variable being equal to $r$.

These classes of operations are all closed under composition.
For any measurement operation
$\mc{M}_{\A''\B''|\A'\B'} \in \mops^{\goc}(\A':\B' \to \A'':\B'')$
and operation
$\mc{W}_{\A'\B'|\A\B} \in \mops^{\goc}(\A:\B \to \A':\B')$
(where $\goc \in \{\LO, \oneway, \PPT, \all\}$)
we have the closure property
\begin{equation}
    \mc{M}_{\A''\B''|\A'\B'} \mc{W}_{\A'\B'|\A\B}
    \in \mops^{\goc}(\A:\B \to \A'':\B'').
\end{equation}
Note, however, that following a measurement operation with an operation
which is not measuring won't necessarily result in a measurement operation.
All the classes mentioned are closed under convex combination.
Furthermore, they form a hierarchy
$\all \supset \PPT \supset \oneway \supset \LO$
\cite{Rains-rig,2001-Rains-SDP}.

\subsection{Quantum Hypothesis Testing with Restricted Measurements.}
\label{sec:hypoTest}
In a classical hypothesis testing problem
(with simple hypotheses and finite sample space)
there are two hypotheses $H_i$, $i \in \{0,1\}$, of the form
is `the random variable $R$ has distribution $P^{(i)}$'.
A statistical test $T$ can be specified by the giving the
probabilities $T(r) = \Pr( \text{accept } H_0 | T, R = r )$.

The `type-I error' of $T$ is
\begin{align}
    \alpha(P^{(0)},T) =& \Pr(\text{accept }H_1 | H_0,T)\\
    =& 1 - \sum_r P^{(0)}(r) T(r)
\end{align}
while the `type-II error' of $T$ is
\begin{align}
    \beta(P^{(1)},T) =& \Pr(\text{accept }H_0 |H_1,T)\\
    =& \sum_r P^{(1)}(r) T(r).
\end{align}
\begin{definition}\label{c-beta}
\begin{align}
    \rBeta{}{\epsilon}{P^{(0)}}{P^{(1)}}
    :=& \min \beta(P^{(1)},T)\\
    &\text{subject to}\notag\\
    &\alpha(P^{(0)},T) \leq \epsilon,\\
    \forall r:~&0 \leq T(r) \leq 1.
\end{align}
\end{definition}

In a quantum hypothesis testing problem (with simple hypotheses)
there are two hypotheses $H_i$, $i \in \{0,1\}$, of the form
is `the state of system $\Q$ is $\tau^{(i)}_{\Q}$'.
In order to distinguish between these situations it is necessary
to perform a measurement on $\Q$, the outcome of which is
then subjected to a classical hypothesis test.

If the measurement operation is $\mc{M}_{\C|\Q}$
and the classical test has probability $T(r)$ of accepting
when the outcome is $r$, then the overall probability of
acceptance when the state of $\Q$ is $\tau_{\Q}$ is
\begin{equation}
    \sum_{r} T(r) \tr_{\C} \ketbra{r}{r}_{\C} \mc{M}_{\C|\Q} \tau_{\Q}
    = \tr_{\Q} T_{\Q} \tau_{\Q}
\end{equation}
where $T_{\Q} = \mc{M}^{\dagger}_{\Q|\C}
\left( \sum_{r} T(r) \ketbra{r}{r}_{\C} \right)$
is the POVM element corresponding to acceptance of $H_0$.
To see that it is a POVM element, we can use the fact that
$0 \leq \sum_{r} T(r) \ketbra{r}{r}_{\C} \leq 1$,
and that the adjoint map $\mc{M}_{\Q|\C}^{\dagger}$
of an operation is completely positive and unital.

With no restriction on the measurement operation
(i.e. for class $\all$), $T_{\Q}$ can be any valid POVM
element. Therefore, we are justified in defining
\begin{definition}
\begin{align}
    \rBeta{\all}{\epsilon}{\tau_{\Q}^{(0)}}{\tau_{\Q}^{(1)}}
    :=& \min \tr_{\Q} \tau_{\Q}^{(1)} T_{\Q}\\
    &\text{subject to}\notag\\
    &1 - \tr_{\Q} \tau_{\Q}^{(0)} T_{\Q} \leq \epsilon,\\
    \forall r:~&0 \leq T_{\Q} \leq 1.
\end{align}
\end{definition}

Suppose $\tau_{\Q}^{(0)}$ and $\tau_{\Q}^{(1)}$ commute,
and $\ketbra{\eta_i}{\eta_i}_{\Q}$ is a common eigenbasis
for them. Defining the operation
\begin{equation}
    \mc{K}_{\Q|\Q}: X_{\Q} \mapsto \sum_{i=1}^{\dim(\Q)}
    \ketbra{\eta_i}{\eta_i}_{\Q} X_{\Q} \ketbra{\eta_i}{\eta_i}_{\Q}
\end{equation}
which completely dephases the system in this eigenbasis, we have
$\mc{K}_{\Q|\Q} \tau^{(0)}_{\Q} = \tau^{(0)}_{\Q}$ and
$\mc{K}_{\Q|\Q} \tau^{(1)}_{\Q} = \tau^{(1)}_{\Q}$.
Since $\mc{K}_{\Q|\Q}$ is its own adjoint map, this means that
there a test $T_{\Q}$, optimal for the hypothesis test,
which itself satisfies $\mc{K}_{\Q|\Q} T_{\Q} = T_{\Q}$. That is,
it too is diagonal in the common eigenbasis of the two states.
In particular, if both $\tau_{\Q}^{(0)}$ and $\tau_{\Q}^{(1)}$ are
diagonal in the classical basis, then we can assume without loss
of generality that the test is too, whereupon everything reduces to the
classical case.

We extend this definition to restricted measurements
on bipartite systems by
\begin{definition}\label{rmeas}
For $\goc \neq \all$,
\begin{align}
	&\rBeta{\goc}{\epsilon}{\tau^{(0)}_{\A \B}}{\tau^{(1)}_{\A \B}}
	:= \inf
	\rBeta{}{\epsilon}{\mmap \tau^{(0)}_{\A \B} }{\mmap \tau^{(1)}_{\A\B} }\\
    &\text{subject to}\notag\\
    & \mmap \in \mops^{\goc}(\A:\B \to \A':\B'),\\
    &\text{for arbitrary }\A',\B'.
\end{align}
That is, we reduce the quantum to the classical case by optimising
over all measurement operations belonging to the allowed class $\goc$.
Based on this quantity, we define the
$\goc$-\emph{hypothesis-testing relative entropy}
\begin{equation}
    \rDH{\goc}{\epsilon}{\tau^{(0)}_{\A \B}}{\tau^{(1)}_{\A \B}} :=
- \log \rBeta{\goc}{\epsilon}{\tau^{(0)}_{\A \B}}{\tau^{(1)}_{\A \B}}.
\end{equation}
\end{definition}
A trivial but very useful result is the following
\begin{proposition}[Data processing inequality]\label{dpi}
    If $\goc$ is closed under composition, then
	for any operation $\mathcal{N}_{\A'\B'\qcon\A\B} \in
	\ops^{\goc}(\A:\B \to \A':\B')$
	\begin{equation}
		\rDH{\goc}{\epsilon}
		{\mathcal{N}_{\A'\B'\qcon\A\B} \tau^{(0)}_{\A \B} }
		{\mathcal{N}_{\A'\B'\qcon\A\B} \tau^{(1)}_{\A \B} }
		\leq \rDH{\goc}{\epsilon}{\tau^{(0)}_{\A \B}}{\tau^{(1)}_{\A \B}}.
	\end{equation}
\end{proposition}
\begin{corollary}\label{revdpi}
	If there is also an operation
	$\mathcal{N}'_{\A\B\qcon\A'\B'} \in \ops^{\goc}(\A:\B \to \A':\B')$
	such that
	$\mathcal{N}'_{\A\B\qcon\A'\B'}
	 \mathcal{N}_{\A'\B'\qcon\A\B} \tau^{(0)}_{\A \B} 
	= \tau^{(0)}_{\A \B}$
	and $\mathcal{N}'_{\A\B\qcon\A'\B'}
	    \mathcal{N}_{\A'\B'\qcon\A\B} \tau^{(1)}_{\A \B}
	= \tau^{(1)}_{\A \B}$ then
	\begin{equation}
		\rDH{\goc}{\epsilon}
		{\mathcal{N}_{\A'\B'\qcon\A\B} \tau^{(0)}_{\A \B}}
		{\mathcal{N}_{\A'\B'\qcon\A\B} \tau^{(1)}_{\A \B}}
		= \rDH{\goc}{\epsilon}{\tau^{(0)}_{\A \B} }{ \tau^{(1)}_{\A \B}}.
	\end{equation}
\end{corollary}

Suppose that the optimal (classical) hypothesis test
which acts on measurement result in Definition \ref{rmeas}
has probability $T(a,b)$ of accepting $H_0$ when 
the measurement operation produces the outcome $\ketbra{a}{a}_{\A'}\ketbra{b}{b}_{\B'}$.
Then the overall probability of the test accepting $H_0$
when the state of $\A\B$ is $\tau_{\A\B}$ is
\begin{align}
    &\sum_{a,b} T(a,b)
    \tr_{\A' \B'}
    \ketbra{a}{a}_{\A'}\ketbra{b}{b}_{\B'}
    \mmap \tau_{\A \B}\\
    =& \tr_{\A\B} T_{\A\B} \tau_{\A\B} 
\end{align}
where (with $\mc{M}^{\dagger}_{\A\B\qcon\A'\B'}$ the adjoint map
for $\mmap$),
\begin{equation}
    T_{\A\B} = \mc{M}^{\dagger}_{\A\B\qcon\A'\B'}
    \left(\sum_{a,b} T(a,b) \ketbra{a}{a}_{\A'}\ketbra{b}{b}_{\B'}\right).
\end{equation}

$T_{\A\B}$ is the POVM element (which we call simply a `test')
corresponding to acceptance of hypothesis $H_0$, in some quantum
hypothesis test which can be implemented by a measurement operation
in $\goc$ followed by a classical hypothesis test in the joint outcome.
We denote the set of such tests on $\A\B$ by by $\tests^{\goc}(\A:\B)$.

As shown in \cite{VP},
$\tests^{\PPT}(\A:\B)$ consists of all POVM elements $T_{\A\B}$
(that is, $0 \leq T_{\A\B} \leq \1_{\A\B}$), satisfying
\begin{equation}
    0 \leq \transmap_{\B|\B} T_{\A\B} \leq \1_{\A\B}.
\end{equation}
Working through the definitions, we see that
$\tests^{\LO}(\A:\B)$
is the convex hull of all POVM elements of the form
\begin{equation}\label{LOT}
    \sum_{a,b} T(a,b) E(a)_{\A} D(b)_{\B}
\end{equation}
where $\{ E(a)_{\A} \}$ is a POVM on $\A$,
$\{ D(b)_{\B} \}$ is a POVM on $\B$,
and $0 \leq T(a,b) \leq 1$.
$\tests^{\oneway}(\A:\B)$
is the convex hull of all POVM elements of the form
\begin{equation}\label{onewayT}
    \sum_{a,b} T(a,b) E(a)_{\A} D^{a}(b)_{\B}
\end{equation}
where $\{ E(a)_{\A} \}$ is a POVM on $\A$,
and for each $a$, $\{ D^{a}(b)_{\B} \}$
is a POVM on $\B$, and $0 \leq T(a,b) \leq 1$.
(We can use the results
of \cite{2005-DArianoLoPrestiPerinotti} to show that,
given our assumption that $\A$ and $\B$ are finite dimensional,
it suffices to take local POVMs with a finite number of outcomes
in these last two statements.)

\begin{proposition}\label{test-beta}
    \begin{align}
	\rBeta{\goc}{\epsilon}{\tau^{(0)}_{\A \B}}{\tau^{(1)}_{\A \B}}
	=& \inf \tr \tau^{(1)}_{\A \B} T_{\A \B} \\
	&\text{subject to}\notag\\
	&\tr_{\A\B} \tau^{(0)}_{\A\B} T_{\A \B} \geq 1 - \epsilon,\\
	&T_{\A \B} \in \tests^{\goc}(\A:\B).
    \end{align}
\end{proposition}

\begin{remark}
	For the classes of operations considered in this
	paper ($\all$, $\PPT$, $\oneway$ and $\LO$),
	$\tests^{\goc}(\A:\B)$ is a closed set, and so for these
classes
	the infimum in Proposition \ref{test-beta} and Definition
	\ref{rmeas} can be replaced by a minimum.
\end{remark}

\subsection{Codes}\label{codes}
As usual, a use (or uses) of a quantum channel with input system $\A$
and output system $\B$ is represented by an operation
$\chS \in \ops(\A\to\B)$.

\begin{definition}
    \label{def:ea-code}
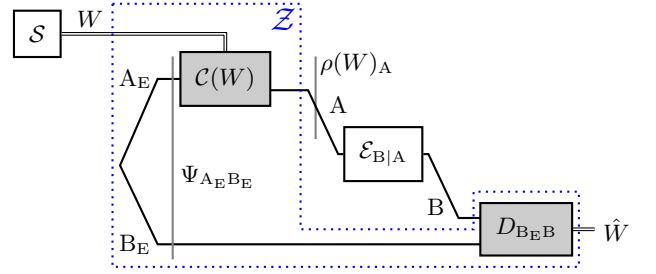
\begin{figure}[h]
	\begin{tikzpicture}[scale=1]
		\node (abox) at (-3,1) {};
		
		\node[rectangle, draw, thick, fill=white, inner sep=2mm]
		(chan) at (-0.9,0) {$\mc{E}_{\outS \qcon \A }$};
		
		\node (bbox) at (1,-1) {};

		\node (output) at (2.2,-1) {$\hat{W}$};

		\node[rectangle, draw, thick, fill=white, inner sep=2mm]
		(src) at (-5.5,1.6) {$\src$};
		
		\draw[draw=black, style=thick]
		(-3,1)--(-3.9,1) -- (-4.4,-0.15) -- (-3.9,-1.2) -- (1,-1.2);

		\draw[draw=black, style=gray,thick]
		(-1.8,1.3) -- (-1.8,0.2);
		\node (encstate) at (-1.25,1.2) {$\rho(W)_{\A}$};
		
		\node (input) at (-4.8,1.8) {$W$};
		\draw[draw=black, style=double]
		(src) -| (-3,1.2);

		\draw[draw=black, style=thick]
		(-3,0.85)--(-1.9,0.85)--
		(-1.5,0)--(chan)--(-0.3,0)--
		(0.1,-0.85)--(1,-0.85);
		
		\draw[draw=black, style=gray,thick]
		(-3.7,-1.4) -- (-3.7,1.3);
		\node (ie) at (-3.1,-0.3) {$\ient_{\eA\eB}$};
		
		\draw[draw=black, style=double]
		(bbox) -- (output);
		
		\node[rectangle, draw, thick, fill=black!20, inner sep=2mm]
		(abox2) at (abox) {$\mathcal{C}(W)$};
		
		\node[rectangle, draw, thick, fill=black!20, inner sep=2mm]
		(bbox2) at (bbox) {$D_{\eB\outS }$};

		\draw[draw=black, style=dotted,thick,blue]
		(-4.5,2)--(-2,2)--(-2,-1)--(0.3,-1)--(0.3,-0.5)
		--(1.7,-0.5)--(1.7,-1.5)--(-4.5,-1.5)--(-4.5,2);
		\node[style=blue] (code) at (-2.3,1.8) {
		\begin{large} $\code$ \end{large}};

		\node (AGS) at (-4.2,1) {$\eA$};
		\node (AGcS) at (-4.2,-1.2) {$\eB$};
		\node (AS) at (-1.5,0.65) {$\A$};
		
		\node (BS) at (-0.2,-0.7) {$\outS$};
	\end{tikzpicture}
  \centering
  \caption{
  An entanglement-assisted code $\code$ transmitting a message
  $W$ produced by a source $\src$, via a channel use $\mc{E}$.
  The average channel input induced by the source and encoding
  is $\iS = \sum_{w = 1}^M \src(w) \rho(w)_{\A}$.
  }\label{fig:EA}
\end{figure}
  In an {\bf entanglement-assisted code}
  of size $M$, the sender and receiver have
  systems $\eA$ and $\eB$ in an entangled state $\ient_{\eA \eB}$,
  and for each message $w \in \{1,2, \ldots, M\}$
  there is an encoding operation
   $\enc(w)_{\A\qcon\eA} \in \ops(\eA\to\A)$. 
  Following the use(s) of the channel,
  the decoder performs a POVM $D_{\eB \outS }$ on $\eB \outS$ to obtain the
  decoded message. $D(\hat{w})_{\eB \outS }$ is the POVM element 
  corresponding the decoded message being $\hat{w}$.
\end{definition}

\begin{definition}
\label{def:ua-code}
An {\bf unassisted code} can be viewed as a degenerate case
    of an entanglement-assisted code
  where the decoding measurement operates
  only on the channel output $\outS$. Since $\eB$ is completely ignored,
  there is no loss of generality if we take $\eB$ and $\eA$ to be trivial,
  one-dimensional systems, so that $\enc(w)_{\A\qcon\eA}$
  is completely specified by its output $\rho(w)_{\A}$ on $\A$. 
\end{definition}

Figure \ref{fig:EA} illustrates an entanglement-assisted code $\code$
transmitting a message $W$ produced by a source $\src$ via a channel use
$\chS$. The message $W$ and the outcome $\hat{W}$ of the decoding POVM
are classical random variables. The source is specified by the probabilities
\begin{equation}
    \src(w) := \Pr(W=w|\src).
\end{equation}
The probability of error (which depends
on the source, code and channel) is
\begin{equation}
    \Pr(\dmRV\neq\mRV|\chan,\code,\src).
\end{equation}
For $M \in \mathbb{N}$ let $\src_{M}$ denote a source with
$M$ equiprobable messages
i.e.
$\src_{M}(w) = 1/M$.

\begin{definition}\label{def:code-size}
    We call a size $M$ code $\code$
	an $(M,\epsilon,\iS)$ code for $\mc{E}$,
    if
    its average error probability for $M$ equiprobable messages satisfies
	$\Pr(\dmRV\neq\mRV|\chan,\code,\src_{M}) \leq \epsilon$,
	and the average channel input state induced by using
	the code with equiprobable messages is $\iS$.

    We denote by
$\Mi{\epsilon}{\chS}{\iS}$
the largest $M$ such that there is an
$(M,\epsilon,\iS)$ unassisted code,
and by $\M{\epsilon}{\chS}$ the largest $M$
such that there exists an $\iS$ such that there is an $(M,\epsilon,\iS)$
unassisted code.

$\Mei{\epsilon}{\chS}{\iS}$ and $\Me{\epsilon}{\chS}$ denote
the corresponding quantities for entanglement-assisted codes.

\end{definition}
\begin{remark}
    Clearly,
    \begin{equation}
        \Me{\epsilon}{\chS} = \displaystyle\max_{\iS \in \states(\A)} \Mei{\epsilon}{\chS}{\iS}
    \end{equation} and
    \begin{equation}
        \M{\epsilon}{\chS} = \displaystyle\max_{\iS \in \states(\A)} \Mi{\epsilon}{\chS}{\iS}
    \end{equation}
\end{remark}

\section{Summary of results}\label{summary}

In \cite{PPV} Polyanskiy, Poor and Verd\'{u}
showed that many existing \emph{classical}
converse results can be easily
derived from a finite blocklength converse 
(Theorem 27 of \cite{PPV}) which we will call the \emph{PPV converse}.
It is obtained by a simple and conceptually
appealing argument relating coding to
hypothesis testing on the joint distribution of the
channel input and channel output.
Our bounds are given in terms of a quantum hypothesis
testing problem on a bipartite system, consisting of
$\B$ (the output system for the channel operation)
and $\At$ (a copy of the input system). To compactly
describe the hypotheses, we first introduce a little
notation.

\begin{figure}[ht]
    \begin{tikzpicture}[scale=1.2]
        \node (h0) at (-2.7,0.3) {$H_0:$};
		\node[rectangle, draw, thick, fill=white, inner sep=2mm]
		(chan) at (0,0) {$\chan$};
		\draw[draw=black, style=thick]
		(2,0.6) -- (-2,0.6) -- (-2.2,0.3) -- (-2,0) -- (chan) -- (2,0);
		\draw[draw=black, style=gray]
		(1.2,1) -- (1.2,-0.4);
		\node (out) at (1.4,-0.5) {$\chS \purho$};
		\draw[draw=black, style=gray]
		(-1.8,1) -- (-1.8,-0.4);
		\node (out) at (-1.8,-0.5) {$\purho$};
		\draw[draw=black, style=gray]
		(-1,0.1) -- (-1,-0.4);
		\node (out) at (-1,-0.5) {$\rho_{\A}$};
		\draw[draw=black, style=gray]
		(-1,0.5) -- (-1,0.85);
		\node (outc) at (-1,1) {$\rho_{\At}$};
		\node (AcS) at (-1.6,0.8) {$\At$};
		\node (AcS2) at (0.8,0.8) {$\At$};
		\node (AS) at (-1.6,0.2) {$\A$};
		\node (BS) at (0.8,0.2) {$\outS$};
	\end{tikzpicture}
	\begin{tikzpicture}[scale=1.2]
        \node (h) at (-2.7,0.3) {$H_1:$};
		\node[rectangle, draw, thick, fill=white, inner sep=2mm]
		(chan) at (.2,0) {$\sigma_{\outS}$};	
		\node [rectangle, draw, thick, fill=black, inner sep=2mm]
		(deadend) at (-0.5,0) {};
		\draw[draw=black, style=thick]
		(2,0.6) -- (-2,0.6) -- (-2.2,0.3) -- (-2,0) -- (deadend);
		\draw[draw=black, style=thick]
		(chan) -- (2,0);
		
		\draw[draw=black, style=gray]
		(1.2,1) -- (1.2,-0.4);
		\node (out) at (1.4,-0.5) {$\rho_{\At}\sigma_{\outS}$};
		\draw[draw=black, style=gray]
		(-1.8,1) -- (-1.8,-0.4);
		\node (out) at (-1.8,-0.5) {$\purho$};
		\draw[draw=black, style=gray]
		(-1,0.1) -- (-1,-0.4);
		\node (out) at (-1,-0.5) {$\rho_{\A}$};
		\draw[draw=black, style=gray]
		(-1,0.5) -- (-1,0.85);
		\node (outc) at (-1,1) {$\rho_{\At}$};
		\node (AcS) at (-1.6,0.8) {$\At$};
		\node (AcS2) at (0.8,0.8) {$\At$};
		\node (AS) at (-1.6,0.2) {$\A$};
		\node (BS) at (0.8,0.2) {$\outS$};
	\end{tikzpicture}
  \centering
  \caption{The quantum hypothesis testing problem which appears in our bounds.
}\label{fig:qH}
\end{figure}
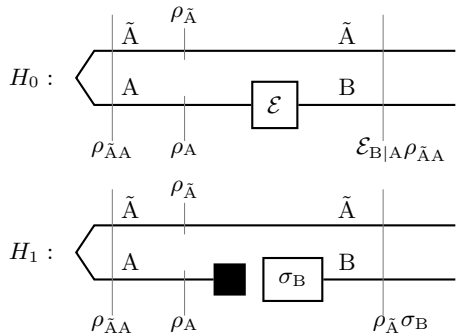   

\begin{definition}\label{Phi-def}
    For any two systems $\Q$, $\Qt$ of the same dimension, let
    $\Phi_{\Qt\Q} := \sum_{i,j=1}^{\dim(\Q)}
\ket{i}_{\Qt}\ket{i}_{\Q}
\bra{j}_{\Qt}\bra{j}_{\Q}$.
\end{definition}
\begin{remark}\label{transpose-trick}
    $\Phi_{\Qt\Q}$ is $\dim(\Q)$ times the isotropic,
    maximally entangled state of $\Qt\Q$.
    Note that $\tr_{\cpy{Q}} \Phi_{\Qt\Q} = \1_{\Q}$
    and $\tr_{\Q} \Phi_{\Qt\Q} = \1_{\Qt}$.
    We will make use, more than once, of the basic fact that
    for any linear operator
$M_{\Q}$ on $\mc{H}_{\Q}$, $M_{\Q} \Phi_{\Qt\Q}
    = M^{\T}_{\Qt} \Phi_{\Qt\Q}$
\end{remark}
\begin{definition}\label{sch-def}
    Given a state $\rho_{\A}$,
    let $\At$ be a copy of system $\A$,
    and define a canonical purification of $\rho_{\A}$ on
    $\A \At$ by
    \begin{equation}
        \purho := \iS^{\half} \Phi_{\At\A} \iS^{\half}.
    \end{equation}
\end{definition}
\begin{remark}
    By Remark \ref{transpose-trick},
    we have $\tr_{\At} \purho = \rho_{\A}$,
    and we find that the
    marginal state of $\At$ is
    \begin{equation}
        \rho_{\At} := \tr_{\A} \purho
        = \id_{\At|\A} \rho_{\A}^{\T}.
    \end{equation}
    Throughout the paper, we regard $\purho$ and $\rho_{\At}$
    as functions of $\rho_{\A}$.
\end{remark}

As shown in Fig.~\ref{fig:qH} the hypotheses specify quantum states
of a bipartite system $\bS$, where $\outS$ is the output
system of $\chS$ and $\At$ is isomorphic to
its input system.
Hypothesis $H_0$ is that $\bS$ is in the state
$\chS \purho$,
whereas hypothesis $H_1$ is that $\bS$ is in the product state
$\rho_{\At}\oS$.

Our main converse bounds are quantum generalisations of Theorem 27
of \cite{PPV}, for entanglement-assisted, and for unassisted, codes.
To state them, let us introduce
\begin{definition}
    \begin{equation}
        \rBch{\mclass}{\epsilon}{\chS}{\iS}
        := \max_{\oS \in \states(\B)}
        \rBeta{\mclass}{\epsilon}{\chS \purho}{\icS \oS},
    \end{equation}
    which is given in terms of Definition \ref{rmeas}.
    In words, this is the minimum type-II error of all tests in
    class $\mclass$ which have type-I error no greater than $\epsilon$
    for the hypothesis testing problem depicted in Figure \ref{fig:qH},
    maximised over all $\sigma_{\outS}$.
\end{definition}

The relationship between the mutual information and
standard relative entropy, informs
\begin{definition}\label{htI}
\begin{align}
    \rIHch{\mclass}{\epsilon}{\chS}{\iS}
    :=& \min_{\oS \in \states(\B)}
    \rDH{\mclass}{\epsilon}{\chS \purho}{\icS \oS}
    \\=&-\log \rBch{\mclass}{\epsilon}{\chS}{\iS}.
\end{align}
We note again that $\purho$ is here the canonical
pure state of Definition \ref{sch-def}
and is to be regarded as a function of $\rho_{\A}$.
\end{definition}

With these and Definition \ref{def:code-size},
we can now state the main results of this paper in a compact form:
\begin{theorem}[Entanglement-assisted converse]
    \label{se-converse}
\begin{align}
    \log \Mei{\epsilon}{\chS}{\iS}
    \leq&
    \rIHch{\all}{\epsilon}{\chS}{\iS}, \text{ and so}\\
    \log \Me{\epsilon}{\chS}
    \leq& \max_{\iS \in \states(\A)}
    \rIHch{\all}{\epsilon}{\chS}{\iS}.
\end{align}
\end{theorem}
\begin{theorem}[Unassisted converse]
    \label{ua-converse}
    For any class of operations $\mclass$ containing $\LO$, we have
\begin{align}
    \log \Mi{\epsilon}{\chS}{\iS}
    \leq&
    \rIHch{\mclass}{\epsilon}{\chS}{\iS}, \text{ and so}\\
    \log M_{\epsilon}(\chS)
    \leq& \max_{\iS \in \states(\A)}
    \rIHch{\mclass}{\epsilon}{\chS}{\iS}.
\end{align}
\end{theorem}
The proof of both of these is given in the next section
(\ref{sec:fbc}). In section \ref{sec:props}
we prove a number of properties of the bounds, which we first summarise here.
The converse for entanglement-assisted codes in Theorem \ref{se-converse}
has a number of desirable properties:
\begin{enumerate}
	\item
	Like the bound of Datta and Hsieh~\cite{DH},
it is \emph{asymptotically tight} for memoryless channels.
That is to say, analysing the large block length behaviour of the bound
    for memoryless channels recovers the converse part of
    the single-letter formula for entanglement-assisted capacity
    proven by Bennett, Shor, Smolin and Thapliyal \cite{BSST}.
This is shown in subsection \ref{ssec:asymp}.

    \item For a fixed blocklength, the converse of
    Datta and Hsieh \cite{DH} grows slowly, but without bound,
    as $\epsilon \to 0$ whereas our converse is
    a decreasing function of $\epsilon$.

\item Generalising results of Polyanskiy~\cite{2013-Polyanskiy}
we show that
    $\rBeta{\all}{\epsilon}{\chS \purho}{\icS \oS}$ is convex in
    $\rho_{\A}$ and concave in $\oS$.
    This enables one to use symmetries of the channel to restrict
    the optimisation over $\iS$ and $\oS$ to states with
    corresponding symmetries, as we show in subsection \ref{sym}.

\item In subsection \ref{ssec:sdps}, we give
an explicit formulation of the bound as \emph{semidefinite program} (SDP)
which is a natural generalisation of the linear program (LP)
given in \cite{2012-Matthews} for the PPV converse.
\end{enumerate}

Regarding our converse for unassisted codes, Theorem \ref{ua-converse},
in subsection \ref{ssec:WR}
the Wang-Renner bound is shown to be equivalent to
making the (sometimes suboptimal \cite{2013-Polyanskiy}) choice
$\oS = \chS \iS$
and taking $\mclass$ to be the class of operations $\oneway$
(local operations and one-way classical communication from Alice to Bob).

Since the Wang-Renner bound is asymptotically tight for the
unassisted capacity (and even for
the product state capacity, thus recovering the HSW theorem),
the stronger bound obtained using $\LO$
also has these properties. Unfortunately, it lacks an SDP
formulation and does not possess the concavity property
mentioned above.
However, the formulation in terms of restricted hypothesis testing makes
it clear that by moving to less restrictive conditions on the test, we might
obtain weaker, but more tractable bounds.
When $\mclass$ is $\oneway$, or the larger class $\PPT$ (see next section),
the concavity property does hold (see Theorem \ref{thm:cvx}),
and we can therefore use the symmetrisation arguments.

For $\PPT$ the bound has the advantage that it can be formulated as
an SDP (see subsection \ref{ssec:sdps}).
It seems unlikely that the $\PPT$ bound is in general,
asymptotically tight, but it might prove useful for certain channels.

\section{Proof via metaconverse}\label{sec:fbc}

Just as in \cite{PPV}, our main results (Theorems \ref{se-converse}
and \ref{ua-converse}) are consequences of a
more general `meta-converse'.
Following \cite{2010-PolyanskiyVerdu,2012-SharmaWarsi}
we first express the general idea of the meta-converse
using ``generalised divergences''. The hypothesis testing
based bounds can be obtained by using the hypothesis-testing
relative entropies as the divergences.

Let $(1-\epsilon, \epsilon)_{\C}$ denote
a state diagonal in the classical basis
with eigenvalues $(1-\epsilon, \epsilon)$
on a two-dimensional system $\C$
(it represents a binary probability distribution).

Suppose that, given a code $\code$, we can find
a measurement operation $\bmo_{\C|\A\B}$ with a binary outcome such that,
for all channel operations $\chS$ from $A$ to $B$, we have
\begin{equation}\label{bid}
    \bmo_{\C|\A\B} \chS \purho = (1-\epsilon, \epsilon)_{\C},
\end{equation}
where $\epsilon$ is the error probability of the
code $\mc{Z}$ for $\mc{E}$,
and $\purho$ is a canonical purification of the average
input $\rho_{\A}$ made to the channel
by the encoder (see Definition \ref{sch-def}).

Taking a reference channel operation $\mc{F}_{\outS|\A}$
for which the error probability $f$ obtained by the
code is known (or bounded) and
a measure of state distinguishability 
(a ``generalised divergence'') $\sdm$
which is non-increasing under the operation $\bmo$
we have
\begin{align}
    &\sdm((1-\epsilon, \epsilon)_{\C}, (1-f, f)_{\C})\\
    &\qquad =\ \sdm(\bmo_{\C|\A\B} \chS \purho ,
    \bmo_{\C|\A\B} \mc{F}_{\outS|\A} \purho )\\
    &\qquad\leq\ \sdm(\chS \purho,
     \mc{F}_{\outS|\A} \purho),
\end{align}
bounding the difference between $\epsilon$ and the known
quantity $f$. 

In the classical meta-converse of Polyanskiy et al. \cite{PPV}
$\rho_{\A} = \sum_{x} p(x) \ketbra{x}{x}_{\A}$ is the
classical distribution on the input alphabet
induced by the encoding of equiprobable messages, and
$\chS \purho$ is the joint probability distribution
over the input and output alphabets induced by the
use of the channel.
In Wang--Renner and in Hayashi \cite{2006-HayashiBook},
$\chS$ is a classical-quantum (c-q) channel which takes
input symbol $x$ to some output state $\tau(x)_{\outS}$.
In these bounds $\rho_{\A}$ is again
a probability distribution over the input symbols,
while $\chS \purho$
is now the quantum state
\begin{equation}
    \chS \purho = \sum_{x} p(x)\ketbra{x}{x}_{\At} \otimes
     \tau(x)_{\B}
\end{equation}
where $\tau(x)_{\B}$ is the output of the classical-quantum channel
on the input symbol $x$. 
If a code of size $M$ is used with a channel operation whose output
is a fixed state $\sigma$, independent of its input, its error
probability is $1/M$.
Hayashi and Wang-Renner implicitly use an such an operation
for $\mc{F}_{\outS|\A}$ with
$\sigma_{\outS} = \chS \rho_{\A}$.
In Theorem 27 of Polyanskiy et al. \cite{PPV} the bound is optimised over all such operations.
Hayashi uses the quantum relative R\'{e}nyi entropy for
$\sdm$ while Wang-Renner and Polyanskiy et al. use
the (unrestricted) hypothesis-testing relative entropy,
which itself depends on $\epsilon$. 

We first show how to construct from an entanglement-assisted code
a measurement operation satisfying (\ref{bid}).
\begin{proposition}\label{code2test}
    From any entanglement-assisted code $\code$
    and source $\src$,
    such that the average input state
    is $\rho_{\A}$,
    one can construct a test $T_{\bS} \in \tests(\At:\outS)$
    such that, for all $\chS \in \ops(\A \to \B)$,
    \begin{equation}
    \Pr(\dmRV=\mRV|\chan,\code,\src)
    = \tr_{\bS} T_{\bS} \chS \purho.
    \end{equation}
    Furthermore, if $\code$ is unassisted, then
    $T_{\bS} \in \tests^{\LO}(\At:\outS)$.
\end{proposition}
\begin{proof}
We consider a general entanglement-assisted code as depicted
in Figure \ref{fig:EA}, and
using the notation established in Definition \ref{def:ea-code}.
Since it is always possible to augment $\eA$ to $\eA' \eA$, 
let $\ient_{\eA'\eA\eB}$ be a purification of $\ient_{\eA\eB}$,
and take $\enc(\mV)_{\A\qcon\eA'\eA}
:=  \enc(\mV)_{\A\qcon\eA} \tr_{\eA'}$,
we can assume that $\ient_{\eA\eB}$ is pure.

Let the isometry
$U(w) \in \lm(\Hs_{\eA},\Hs_{\sT}\ox\Hs_{\A})$
be the Stinespring representation
of the encoding map $\enc(\mV)_{\A\qcon\eA}$,
where $\sT$ is the discarded
environment system. In fact, we can just take $\eA = \sT\A$ so that
$U(w)_{\sT\A}$ is a unitary.
So, the encoding map for message $w$ is
$\enc(\mV)_{\A\qcon\sT\A}: X_{\sT\A} \mapsto
\tr_{\sT} U(w)_{\sT\A} X_{\sT\A} U(w)_{\sT\A}^{\dagger}$.
Finally, there is no loss of generality in demanding that
$\eB = \cpy{\sT}\At \cong \sT\A$.
This reformulation of the protocol of Fig. \ref{fig:EA}
is illustrated in Fig. \ref{fig:EA2}.

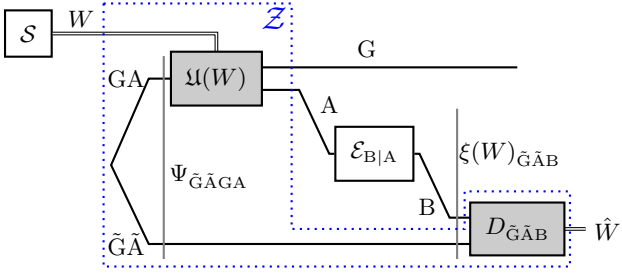
\begin{figure}[h]
	\begin{tikzpicture}[scale=1]
		\node (abox) at (-3,1) {};
		
		\node[rectangle, draw, thick, fill=white, inner sep=2mm]
		(chan) at (-0.9,0) {$\mc{E}_{\outS \qcon \A }$};
		
		\node (bbox) at (1,-1) {};

		\node (output) at (2.2,-1) {$\hat{W}$};

		\node[rectangle, draw, thick, fill=white, inner sep=2mm]
		(src) at (-5.5,1.6) {$\src$};
		
		\draw[draw=black, style=thick]
		(-3,1)--(-3.9,1) -- (-4.4,-0.15) -- (-3.9,-1.2) -- (1,-1.2);
		
		\node (input) at (-4.8,1.8) {$W$};
		\draw[draw=black, style=double]
		(src) -| (-3,1.2);
		
		\draw[draw=black, style=thick]
		(-3,1.15) -- (1,1.15);

		\draw[draw=black, style=thick]
		(-3,0.85)--(-1.9,0.85)--
		(-1.5,0)--(chan)--(-0.3,0)--
		(0.1,-0.85)--(1,-0.85);
		

		\draw[draw=black, style=gray,thick]
		(-3.7,-1.4) -- (-3.7,1.3);
		\node (ie) at (-3.1,-0.3) {$\ient_{\cpy{\sT}\At\sT\A}$};
		
		\draw[draw=black, style=gray,thick]
		(0.2,-1.4) -- (0.2,0.6);
		\node (out) at (0.9,0) {$\xi(W)_{\cpy{\sT}\At\outS}$};
		
		\draw[draw=black, style=double]
		(bbox) -- (output);
		
		\node[rectangle, draw, thick, fill=black!20, inner sep=2mm]
		(abox2) at (abox) {$\mathfrak{U}(W)$};
		
		\node[rectangle, draw, thick, fill=black!20, inner sep=2mm]
		(bbox2) at (bbox) {$D_{\cpy{\sT}\At\outS }$};

		\draw[draw=black, style=dotted,thick,blue]
		(-4.5,2)--(-2,2)--(-2,-1)--(0.3,-1)--(0.3,-0.5)
		--(1.7,-0.5)--(1.7,-1.5)--(-4.5,-1.5)--(-4.5,2);
		\node[style=blue] (code) at (-2.3,1.8) {
		\begin{large} $\code$ \end{large}};

		\node (AGS) at (-4.2,1) {$\sT\A$};
		\node (AGcS) at (-4.2,-1.2) {$\cpy{\sT}\At$};
		\node (AS) at (-1.5,0.65) {$\A$};
		\node (GS) at (-1,1.4) {$\sT$};
		\node (BS) at (-0.2,-0.7) {$\outS$};

	\end{tikzpicture}
  \centering
	\caption{Reformulation of the protocol of Fig. \ref{fig:EA}.}
  \label{fig:EA2}
\end{figure}
First note that
\begin{equation} 
\begin{split}  
    &U(w)_{\sT\A}
	\ient_{\sT\A\cpy{\sT}\At} U(w)_{\sT\A}^{\dag}\\
	=&M(w)_{\sT\A}
	    \ent_{\sT\A\cpy{\sT}\At}
	M(w)_{\sT\A}^{\dagger}
\end{split}
\end{equation}
where $M(w)_{\sT\A} := U(w)_{\sT\A} \ient_{\sT\A}^{\half}$
and $\ient_{\sT\A} := \tr_{\cpy{\sT}\At} \ient_{\sT\A\cpy{\sT}\At}$.
The state of $\sT\A$ after encoding if the message $W = w$ is
\begin{align}
    \rho(w)_{\sT\A} =
    M(w)_{\sT\A} M(w)_{\sT\A}^{\dagger}.
\end{align}
Referring to the diagram, we see that
\begin{equation}
    \begin{split}
	&\Pr(\dmRV=\dmV|\mRV=\mV, \chan, \code)
	\\ =&
	\tr_{\cpy{\sT}\At\outS}
	D(\hat{w})_{\cpy{\sT}\At\outS}
	\xi(w)_{\cpy{\sT}\At\outS}
\end{split}
\end{equation}
where $D(\hat{w})_{\cpy{\sT}\At\outS}$ is the POVM element
corresponding to the decoded message being $\hat{w}$, and
where
\begin{align}
	\xi(w)_{\cpy{\sT}\At\outS}=&
	\tr_{\sT} \chS [
	M(w)_{\sT\A}
	    \ent_{\sT\A\cpy{\sT}\At}
	M(w)_{\sT\A}^{\dagger}
	]\\
	=& \tr_{\sT}
	    M(w)_{\cpy{\sT}\At}^{\T}
		\ent_{\cpy{\sT}\sT} (\chS \ent_{\At\A})
		 M(w)_{\cpy{\sT}\At}^{\ast}\\
	=& M(w)_{\cpy{\sT}\At}^{\T}
		\1_{\cpy{\sT}} (\chS \ent_{\At\A}) 
		 M(w)_{\cpy{\sT}\At}^{\ast}.
\end{align}
Here we have used the facts noted in Remark \ref{transpose-trick}
and that $\ent_{\sT\A\cpy{\sT}\At} = \ent_{\At\A} \ent_{\cpy{\sT}\sT}$.
The probability of successful decoding is
\begin{align}
	&\Pr(\dmRV=\mRV|\chan, \code, \src)\\
	=&
	\sum_{w = 1}^M \src(w)
	\tr_{\cpy{\sT}\At\outS}
	D(w)_{\cpy{\sT}\At\outS}
	\xi(w)_{\cpy{\sT}\At\outS}\\
	=&
	\tr_{\cpy{\sT}\At\outS}
    R_{\cpy{\sT}\At\outS} \chS \ent_{\At\A}
\end{align}
where
\begin{align}\label{Rdef}
	R_{\cpy{\sT}\At\outS} :=&
	\sum_{w = 1}^M
    \src(w)
    M(w)_{\cpy{\sT}\At}^{\ast}
	D(w)_{\cpy{\sT}\At\outS}
	M(w)_{\cpy{\sT}\At}^{\T}.
\end{align}
Since $R_{\cpy{\sT}\At\outS}$
is given by a completely positive
map (with Kraus operators
$\sqrt{\src(w)} M(w)_{\cpy{\sT}\At}^{\ast}$)
acting on $D_{\cpy{\sT}\At\outS}$, which satisfies
$0 \leq D_{\cpy{\sT}\At\outS} \leq \1_{\cpy{\sT}\At\outS}$,
we have
$0 \leq R_{\cpy{\sT}\At\outS} \leq
\rho_{\cpy{\sT}\At} \1_{\outS}$,
where $\rho_{\sT\A} := \sum_{w=1}^M \src(w) \rho(w)_{\sT\A}$
is the average state of $\sT\A$ after encoding,
$\rho_{\cpy{\sT}\At\sT\A} := \rho_{\sT\A}^{1/2}
\ent_{\sT\A\cpy{\sT}\At} \rho_{\sT\A}^{1/2}$
is its canonical purification, and
$\rho_{\cpy{\sT}\At} := \tr_{\sT\A} \rho_{\cpy{\sT}\At\sT\A}$.
Therefore,
\begin{align}
	T_{\bS}
:= \rho_{\At}^{-\half } R_{\bS} \rho_{\At}^{-\half }
\end{align}
satisfies $0 \leq T_{\bS} \leq \1_{\bS}$, and 
\begin{align}
    \Pr(\dmRV=\mRV|\chan,\code,\src)
	=& \tr_{\bS} \rho_{\At}^{\half } (\chS \ent_{\At\A})
	\rho_{\At}^{\half } T_{\bS}\\
    =& \tr_{\bS}
    T_{\bS} \chS \purho
\end{align}
as promised.

As noted in the caption for Figure \ref{fig:EA},
any \emph{unassisted} quantum code corresponds to restricting
Bob's decoding measurement to the output system of the channel,
so that
\begin{equation}
    \forall w \in \{1,\ldots, M\}:~D(w)_{\cpy{\sT}\At\outS }
    = \1_{\cpy{\sT}\At} D(w)_{\outS}.
\end{equation}
Substituting this into (\ref{Rdef}), we see that
\begin{equation}\label{local-T}
    T_{\bS}
    = \sum_{w =1}^M
    E(w)_{\At}
    D(w)_{\outS}
\end{equation}
where the positive operators
\begin{equation}
    E(w)_{\At} :=
    S(w)
(\rho_{\At})^{-\half} 
    \left(
    \tr_{\cpy{\sT}}
    \rho(w)_{\cpy{\sT}\At}
    \right)
    (\rho_{\At})^{-\half} 
\end{equation}
satisfy
\begin{equation}
    \sum_{w =1}^M E(w)_{\At} 
    = 
    \rho^{-\half}_{\At}
    \rho_{\At}
    \rho^{-\half}_{\At}
    \leq \1_{\At}.
\end{equation}
Letting $E(0)_{\At} := \1_{\At}
- \rho^{-\half}_{\At}
    \rho_{\At}
    \rho^{-\half}_{\At}$, the
operators $E(0)_{\At}, \ldots, E(M)_{\At}$ constitute a POVM,
and so, for an unassisted code,
$T_{\bS}$ is a \emph{local} test i.e. $T_{\bS} \in \tests^{\LO}(\At:
\outS)$. The local implementation is simply this: Alice performs a
measurement with POVM elements $E(w)_{\At}$ and Bob performs the decoding
measurement (with POVM elements $D(w)_{\B}$). If their outcomes are equal
then the test accepts, the `accepting' POVM element being given
in equation ($\ref{local-T}$).
\end{proof}


A quantum generalisation of the `meta-converse' Theorem 26 of \cite{PPV},
is now straightforward:
\begin{proposition}[Meta-converse]
    Let $\code$ be an entanglement-assisted code 
which, when used with $\mc{S}$, induces the average input state
$\rho_{\A}$, and which has success probability
\begin{equation}
    \Pr(\dmRV=\mRV|\chan^{(i)},\code,\src) = 1 - \epsilon_i,
\end{equation}
when used with channel operation $\chS^{(i)}$, for $i \in \{0,1\}$.
Consider the hypothesis testing problem where $H_i$ asserts that the state
of $\bS$ is $\chS^{(i)}\purho$. If we accept $H_0$ when
the test constructed
from $(\code,\src)$ as in Proposition \ref{code2test} accepts, then
\begin{equation}
    \beta = \tr_{\bS} T_{\bS} \chS^{(1)}\purho  =  1 - \epsilon_1
\end{equation}
and
\begin{equation}
    1 - \alpha =  \tr_{\bS} T_{\bS} \chS^{(0)}\purho  = 1 - \epsilon_0.
\end{equation}
Therefore, by definition
\begin{equation}
    \rBeta{\all}{\epsilon_0}{\chS^{(0)}\purho}{\chS^{(1)}\purho}
    \leq 1 - \epsilon_1.
\end{equation}
Furthermore, if constraints on the code $\code$ mean that $T_{\bS}$
is guaranteed to belong to some class of tests $\tests^{\goc}(\A:\B)$,
then the (potentially more stringent) bound
\begin{equation}\label{MC-final}
    \rBeta{\mclass}{\epsilon_0}{\chS^{(0)}\purho}{\chS^{(1)}\purho}
    \leq 1 - \epsilon_1.
\end{equation}
also applies.
\end{proposition}

We now apply the meta-converse to obtain our main results.
Suppose that
$\chS^{(1)}$, is a completely useless channel operation which has
$\chS^{(1)} \rho_{\A} = \sigma_{\outS}$
for all $\rho_{\A}$, and assume that
the $M$ messages are equiprobable i.e. $\src = \src_M$.
Then, it is easily verified that $1 - \epsilon_1 = 1/M$ (in fact
any other value would imply that communication is possible in
the absence of a channel.) Setting $\chS^{(0)} = \chS$,
the hypothesis testing problem described in the Proposition above
is now exactly the one shown in Figure \ref{fig:qH}.
With these choices, equation (\ref{MC-final}) tells us that
any $(M,\epsilon,\iS)$ code whose corresponding test belongs
to $\tests^{\goc}(\At:\outS)$ must satisfy
\begin{equation}
    \max_{\oS} \rBeta{\mclass}{\epsilon}{\chS \purho}{\icS\oS} \leq 1/M.
\end{equation}
Here, we maximise over $\oS$ to obtain the best possible bound.
For entanglement-assisted codes, rearranging this
and using Definition~\ref{rmeas}
gives us
\begin{equation}
    \log \Mei{\epsilon}{\chS}{\iS} \leq
    \min_{\oS} D^{\all}_{\epsilon}(\chS \purho\|\rho_{\At} \oS)
\end{equation}
(Theorem \ref{se-converse}) and for unassisted codes,
we can write the stronger bound
\begin{equation}
    \log \Mi{\epsilon}{\chS}{\iS} \leq
    \min_{\oS} D^{\LO}_{\epsilon}(\chS \purho\|\rho_{\At} \oS)
\end{equation}
(which is Theorem \ref{ua-converse}).

\section{Properties of the bounds}
\label{sec:props}

In this section, we generalise results of Polyanskiy~\cite{2013-Polyanskiy},
showing that
$\beta^{\mclass}_{\epsilon}(\chS \purho \| \rho_{\At} \oS)$
is convex in
$\rho_{\A}$ and concave in $\oS$, provided $\mclass$ contains
$\oneway$.
This enables one to use symmetries of the channel to restrict
the optimisation over $\iS$ and $\oS$ to states with
corresponding symmetries, as we show in subsection \ref{sym}.
    
\def\ciS{\At}
Let $\beta(T,\icS,\oS) := \tr_{\bS} T_{\bS} \icS\oS $.
This is a bilinear function
of $T$ and $\oS$.
Therefore, the minimum of $\beta(T,\icS,\oS)$ over $T$ is concave
in $\oS$, and since the set $\states(\B)$
and the set of tests satisfying $\alpha(T,\icS) \leq \epsilon$
are both convex, von Neumann's minimax theorem
tells us that
$\max_{\oS} \beta_{\epsilon}^{\goc}(\chS\purho, \icS\oS)$ is also equal to
\begin{align}
    &\min_{T_{\bS} \in \goc} \beta^{\ast}(T)\\
    &\text{ subject to }\notag\\
    &\tr_{\bS} T_{\bS} \chS\purho  \geq 1 - \epsilon,
\end{align}
where
\begin{align}
    \beta^{\ast}(T,\icS) :=& \max_{\oS} \tr_{\bS} T_{\bS} \icS\oS \\
    =& \| \tr_{\A} T_{\bS} \icS  \|_{\infty}
\end{align}
As noted in \cite{2013-Polyanskiy} for the classical case,
$\beta^{\ast}(T)$ is the (worst case)
probability of type II error for $T_{\bS}$
in the \emph{compound} hypothesis testing problem where
$H_0$ is still that the state is $\chS \purho$, but now
$H_1$ is the compound hypothesis that the state belongs to the
set $\{ \icS\oS : \oS \in \states(\B) \}$.

\begin{theorem}\label{thm:cvx}
    For any operations $\chS^{(0)}$ and $\chS^{(1)}$,
	and class of bipartite operations $\goc$ which contains
	$\oneway$, the function
    \begin{equation}
     (\epsilon,\iS) \mapsto
     \rBeta{\mclass}{\epsilon}
     {\chS^{(0)} \purho}{\chS^{(1)} \purho}
     \end{equation}
    is jointly convex in $\epsilon$ and $\iS$.
\end{theorem}
\begin{proof}
    Suppose that we have,
    for all $j \in \{1,\ldots, m\}$, $\epsilon_j \in [0,1]$,
    $\iS^{(j)} \in \states(\A)$, and $\lambda_j \geq 0$
    such that
    $\sum_{j=1}^m \lambda_j = 1$,
    $\sum_{j=1}^m \lambda_j \rho^{(j)}_{\A} = \rho_{\A}$, and
    $\sum_{j=1}^m \epsilon_j \rho^{(j)}_{\A} = \epsilon$.
    
    Let $T(j)_{\bS}$ be a test in $\tests^{\goc}$
    that achieves
	$\beta_j :=
	\rBeta{\mclass}
	{\epsilon_j}{\chS^{(0)}\purho^{(j)}}{\chS^{(1)}\purho^{(j)}}$.
	That is to say,
	\begin{align}\label{Tj-properties}
        \tr_{\bS} (\1-T(j))_{\bS} \chS^{(0)}\purho^{(j)} \leq& \epsilon_j,\\
        \text{and } \tr_{\bS} T(j)_{\bS} \chS^{(1)}\purho^{(j)} =& \beta_j.
	\end{align}
	The claim is that
	\begin{equation}\label{convexity-claim}
	   \rBeta{\mclass}{\epsilon}{\chS^{(0)} \purho}{\chS^{(1)} \purho}
       \leq \sum_{j=1}^m \lambda_j \beta_j.
	\end{equation}
	
	We shall explicitly construct a test in
	$\goc$ which demonstrates (\ref{convexity-claim}).
	Consider the suboperations
	\begin{equation}
	   \mc{M}^{(j)}_{\At|\At}:
       X_{\At} \mapsto \lambda_j
       \rho^{(j)1/2}_{\At} \rho^{-1/2}_{\At}
       X_{\At}
       \rho^{-1/2}_{\At} \rho^{(j)1/2}_{\At}
	\end{equation}
	for $j \in \{1, \ldots, m\}$ and
	\begin{equation}
	   \mc{M}^{(0)}_{\At|\At}:
       X_{\At} \mapsto
       \Pi_{\At}
       X_{\At}
       \Pi_{\At}
	\end{equation}
	where $\Pi_{\At}$ is the orthogonal projector onto the
	kernel of $\rho_{\At}$.
	The key property of these suboperations is that
	\begin{equation}\label{Mj-action}
	   \mc{M}^{(j)}_{\At|\At} \purho
	   =
	   \begin{cases}
	   \lambda_j
	   \purho^{(j)} \text{ for } j \in \{1, \ldots, m\},\\
	   0 \text{ for } j = 0.
        \end{cases}
	\end{equation}
	Since
	\begin{equation}
	   \Pi_{\At} + \sum_{j=1}^m
	   \lambda_j
	   \rho^{-1/2}_{\At}
	   \rho^{(j)}_{\At} 
	   \rho^{-1/2}_{\At}
	   = \1_{\At},
	\end{equation}
	the suboperations $\{ \mc{M}^{(j)}_{\At|\At} \}_{j=0}^m$
	constitute an instrument.
	
	Suppose that Alice performs this instrument, and classically
	communicates the outcome to Bob. If the outcome is $j$ then
	they perform the test with POVM element $T(j)_{\bS}$
	corresponding to deciding on hypothesis $0$.
	If their decision is stored in a register $\C$,
	then the overall measurement operation is
	\begin{align}
	       \mc{M}_{\C|\bS} =
	       \sum_{j=0}^m& \{
	       \ketbra{0}{0}_{\C} \tr_{\bS} T^{(j)}_{\bS}\\
	       +&
	       \ketbra{1}{1}_{\C} \tr_{\bS} (\1-T^{(j)})_{\bS} \}
	       \mc{M}^{(j)}_{\At|\At}.
	\end{align}	
	Since this implementation uses only one-way classical communication
	followed by operations in class $\goc$, which, by hypothesis,
	contains $\oneway$, this test is also in $\tests^{\goc}(\At:
\outS)$.
	Using (\ref{Mj-action}) and (\ref{Tj-properties}),
	we find that the type I error of this test is
	\begin{equation}
	   \tr_{\C} \ketbra{1}{1}_{\C} \mc{M}_{\C|\bS} \chS^{(0)} \purho
	   = \sum_{j=1}^{m} \lambda_j \epsilon_{j} = \epsilon
	\end{equation}
	while its type II error is
	\begin{equation}
	   \tr_{\C} \ketbra{0}{0}_{\C} \mc{M}_{\C|\bS} \chS^{(1)} \purho
	   = \sum_{j=1}^{m} \lambda_j \beta_{j},
	\end{equation}
	and so the claimed convexity (\ref{convexity-claim}) does hold.
\end{proof}

\begin{corollary}\label{convexity}
   $\rBch{\mclass}{\epsilon}{\chS}{\iS}$ is also jointly convex in
   $\epsilon$ and $\iS$, as it is given by maximising over functions
   with this property.
\end{corollary}

\subsection{Semidefinite programs}\label{ssec:sdps}

As mentioned in the introduction, the converse for entanglement-assisted
codes in terms of unrestricted hypothesis testing, and the converse
for unassisted codes in terms of PPT hypothesis testing, can be formulated
as semidefinite programs. Here we give the first of these in detail, and
describe how to add the PPT constraint.

Letting $R_{\bS} := \icS^{\half} T_{\bS} \icS^{\half}$,
we have 
\begin{align}
    \beta^{\ast}(T,\icS) =& \| \tr_{\A} R_{\bS} \|_{\infty}\\
    =& \min \{ \lambda : \lambda \1_{\outS} \geq \tr_{\A} R_{\bS} \},
\end{align}
and $\alpha(T,\icS) = \tr_{\bS} R_{\bS} \chS \meS $.

\begin{proposition}[Primal SDP]\label{primal}
\begin{align}
   \rIHch{\all}{\epsilon}{\chS}{\iS}
    =& - \log \min_{T,\lambda} \lambda\\
    &\text{subject to}\notag\\
    &\tr_{\A} R_{\bS} \leq \lambda \1_{\outS}, \label{pcon2}\\
    &\tr_{\bS} R_{\bS} \chS \meS   \geq 1 - \epsilon, \label{pcon1}\\
    & R_{\bS} \leq \icS\1_{\outS}, \label{pcon4}\\
    &R_{\bS} \geq 0.
\end{align}

Since the constraints on $\iS$ are semidefinite, the bound
\begin{equation}
    \max_{\iS \in \mc{D}(\mc{H}_{\A})} \rIHch{\all}{\epsilon}{\chS}{\iS}
\end{equation}
from Theorem \ref{se-converse} is also a semidefinite program.
\end{proposition}

\begin{remark}
    The constraint $0 \leq  \transmap_{\outS|\outS} T_{\bS} \leq \1$
	(where $\transmap_{\outS}$ is the transpose map on system $\outS$)
    is equivalent to
    \begin{equation}
        0 \leq \transmap_{\outS|\outS} R_{\bS} \leq \icS\1_{\outS}.
    \end{equation}
    Because the transpose map is linear, adding these constraints on $R_{\bS}$
    to the primal SDP above yields an SDP for
     $\rIHch{\PPT}{\epsilon}{\chS}{\iS}$.
\end{remark}

Associating operators $F_{\bS}$ and $G_{\outS}$
with constraints (\ref{pcon4}) and (\ref{pcon2}),
and a real multiplier $\mu$ with
(\ref{pcon1}) yields the Lagrangian
\begin{align}
    &\lambda + \tr_{\B} G_{\outS}(\tr_{\At} R_{\bS} - \lambda \1_\outS)\\
    +&\tr_{\bS} F_{\bS} (R_{\bS} - \icS \1_{\outS} )\\
    +&\mu (1-\epsilon - \tr R_{\bS} \chS \meS  )\\
    =&\tr_{\bS} R_{\bS} (\1_{\At} G_\outS + F_{\bS} - \chS \meS  )\\
    +&\lambda(1-\tr_{\B} G_\outS)
    + (1-\epsilon)\mu.
\end{align}
from which one can derive the dual SDP.
Below, we show that the optimal value of this dual SDP
is equal to the optimal value of the primal.

\begin{proposition}[Dual SDP]\label{dual}
\begin{align}
    \rIHch{\all}{\epsilon}{\chS}{\iS}
    = -\log \{ \max & (1-\epsilon)\mu - \tr F_{\At}\icS \}
    \label{dobj}\\
    &\text{subject to}\notag\\
    \1_{\At} G_\outS + F_{\bS} &\geq \mu \chS \meS ,\label{dc1}\\
    \tr_{\B} G_\outS &\leq 1,\label{dc2}\\
    G_\outS, F_{\bS}, \mu &\geq 0.\label{dc3}
\end{align}
\end{proposition}
\begin{proof}
    For sufficiently large $a$, the point given by
$G_{\outS} = \1_{\outS}/(2\dim(\outS))$,
$F_{\bS} = a \1_{\bS}$, and any $\mu > 0$
strictly satisfies the dual constraints (\ref{dc1}-\ref{dc3}),
so the
dual SDP is strictly feasible,
and therefore its solution is equal to
the primal solution (see
Theorem 3.1 of \cite{VB-SDP}).
\end{proof}
The maximisation of $(\ref{dobj})$ over states $\iS$ of $\At$
can also be formulated as an SDP, in a similar way to the primal.

\subsection{Classical channels}\label{ssec:classical}
Let $\mc{C}_{\A\qcon \A}$ and
$\mc{C}_{\outS \qcon \outS}$ denote the completely
dephasing operations in the classical bases for
$\A$ and $\outS$, respectively.
A channel operation $\chS$ is classical if
$
    \chS \mc{C}_{\A\qcon \A} = \chS,
$
and
$
    \mc{C}_{\outS \qcon \outS} \chS = \chS,
$
and we can therefore restrict the minimization over
average input states in our bounds to states
$\rho_{\A}$ which are diagonal in the classical basis, thus
\begin{equation}
    \rho_{\A} = \sum_{x} p(x) \ketbra{x}{x}_{\A}.
\end{equation}
Furthermore, since the state 
$
\chS \purho
$
is invariant under
the operation
$\mc{C}_{\outS \qcon \outS} \in \ops^{\LO}(\At\to\At,
\outS\to\outS)$,
for any
class $\goc$ containing $\LO$ we have (by Proposition \ref{dpi})
\begin{align}
	&\rDH{\goc}{\epsilon}{\chS \purho}{\icS \oS}\\
	\geq&
	\rDH{\goc}{\epsilon}{
	\mc{C}_{\outS \qcon \outS} \chS \purho 
	}{
	\mc{C}_{\outS \qcon \outS} \icS \oS}\\
	=&
	\rDH{\goc}{\epsilon}{
	\chS \purho 
	}{
	\icS \oS' }
\end{align}
where $\oS'$ is diagonal in the classical basis
of $\outS$. Therefore, we can also restrict the optimisation
over $\oS$ to classical states and, by the discussion
in subsection \ref{sec:hypoTest}, restrict the test POVM element
to be diagonal in the classical basis. Therefore,
Theorems \ref{se-converse} and \ref{ua-converse} both reduce to
the PPV converse for finite alphabets when the channel is classical.

\subsection{Comparison with Wang--Renner}\label{ssec:WR}
\def\cS{\mathrm{C}} 
\newcommand{\cs}[1]{\ketbra{#1}{#1}_{\cS}} 
\def\ens{\mathrm{Ens}} 

In our notation,
the Wang-Renner converse states that for c-q channels
with finite input alphabet $\mathsf{A}$ and output states
$\tau(x)_{\outS}$ for $x \in \mathsf{A}$
\begin{equation}
    \log M_{\epsilon} \leq
    \sup_{p} \rDH{\all}{\epsilon}{\tau_{\cS\outS}}{\tau_{\cS}\tau_{\outS}}
\end{equation}
where $\cS$ is a system of dimension $|\mathsf{A}|$ and
$\tau_{\cS\outS}
:= \sum_{x \in \mathsf{A}} p(x) \ketbra{x}{x}_{\cS} \ox \tau(x)_{\outS}$.
To apply their converse to general channels, one notes that any
(unassisted) code for a general quantum channel, induces
a c-q channel by its specification of the input states used in the code.
Together with the choice of $p$ the yields a c-q state
\begin{equation}
    \tau_{\cS\outS}
    = \sum_{x} p(x) \ketbra{x}{x}_{\cS} \ox \chS[\rho(x)_{\A}]
\end{equation}
Optimising the Wang-Renner bound over all choices of
input states and distributions $p$ such that the average
channel input to the quantum channel is $\rho_{\A}$,
one obtains
\begin{equation}
    M_{\epsilon}(\chS) \leq \chi_{\epsilon}(\chS, \iS)
\end{equation}
where
\begin{definition}
\begin{equation}
    \chi_{\epsilon} (\chS,\iS) :=
    \max_{\eta_{\cS \A} \in \ens(\rho_{\A})}
        \rDH{\all}{\epsilon}{\chS \eta_{\cS \A}}{\eta_{\cS} \chS \rho_{\A}},
\end{equation}
where
$\ens (\rho_{\A})$
    be the set of all states of the form
    \begin{equation}
        \sum_{k=1}^{\dim(\cS)} p_k \cs{k} \ox \rho(k)_{\A}
    \end{equation}
    where $\cS$ is some finite dimensional system 
	(acting as a classical register) and where
    $p_k \geq 0$, $\sum_k p_k = 1$,
    $\rho(k)_{\A} \in \states(\A)$, and
        $\sum_{k} p_k \rho(k)_{\A} = \rho_{\A}$.
\end{definition}
This notation is motivated by the fact that the Holevo bound
\cite{holevo1973} is given by
$
    \chi (\chS,\iS)
    = \max_{\eta_{\cS \A} \in \ens(\rho_{\A})}
        D(\chS \eta_{\cS \A} \| \eta_{\cS} \chS \rho_{\A}),
$
where $D$ is the usual quantum relative entropy \cite{2001-SchumacherWestmoreland}.
\begin{proposition}
    \begin{equation}
    \chi_{\epsilon}(\chS,\iS)
    = \rDH{\oneway}{\epsilon}{\chS \purho}{\icS \chS \iS}
    \end{equation}
\end{proposition}
\begin{proof}
Let Alice's measurement have the POVM elements $F(k)_{\At}$
where $k$ labels the outcome which she sends to Bob.
There is no loss of generality in having Alice perform
her measurement
before Bob does anything, and storing the outcome
in a classical register $\cS$ to which Bob has access.

Let $p(k) := \tr \icS F(k)_{\At}$ be
the probability of outcome $k$ and
$\rho(k)_{\A}
= \tr_{\At} (\iS^{\half}
\Phi_{\At\A} \iS^{\half}) F_{\At}(k)$.
Under hypothesis 0 the state of $\cS\outS$ is $\chS \eta_{\cS\A}$,
where 
\begin{equation}
    \eta_{\cS\A} = \sum_{k} p(k) \ketbra{k}{k}_{\cS} \ox\rho(k)_{\A},
\end{equation}
while under hypothesis 1 the state is $\eta_{\cS} \chS \eta_{\A}$.
Clearly $\eta_{\A} = \sum_{k} p(k) \rho(k)_{\A} = \rho_{\A}$. 
Therefore,
\begin{align}
    &\rDH{\oneway}{\epsilon}{\chS \purho}{\icS \chS \iS}\\
    =&
    \max_{\eta_{\cS \A} \in \ens(\rho_{\A})}
        \rDH{\all}{\epsilon}{\chS \eta_{\cS \A} }{ \eta_{\cS} \chS \iS }\\
    =& \chi_{\epsilon}(\chS,\iS).
\end{align}
\end{proof}
\begin{corollary} From the above proposition and
    the definitions of the quantities involved,
    the inequalities
    \begin{equation}
    \rIHch{\LO}{\epsilon}{\chS}{\iS} \leq
    \rIHch{\oneway}{\epsilon}{\chS}{\iS} \leq
    \chi_{\epsilon}(\chS,\iS)
    \end{equation}
    follow immediately.
\end{corollary}
\subsection{Asymptotics}\label{ssec:asymp}
It was already noted in \cite{RennerWang} that the
(asymptotically tight)
Holevo bound on unassisted codes  
can be recovered from an asymptotic analysis of the Wang-Renner bound,
which our bounds on unassisted codes subsume
(in fact, an argument of \cite{DMHB} can be used to show
that the converse part of the HSW theorem \cite{HSW-H,HSW-SW}
can also be derived).

For entanglement-assisted coding over memoryless quantum channels,
Shannon's noisy channel coding theorem has a beautiful generalisation due to
Bennett, Shor, Smolin and Thapliyal:
\begin{align}
    \lim_{\epsilon \to 0}
    \lim_{n \to \infty}
    \frac{1}{n} \log \Me{\epsilon}{(\chan^{\ox n})_{\outS^n \qcon \A^n}}
    = \max_{\rho_{\A}} I(\chS,\iS)
\end{align}
where
\begin{equation}
    I(\chS,\iS) := S(\iS) + S(\chS(\iS)) - S(\chS \purho)
\end{equation}
is the \emph{quantum mutual information} between systems $\At$
and $\outS$ when the state of $\At\outS$ is $\chS \purho$.
As noted, for classical channels Theorem \ref{se-converse} reduces
to Theorem 27 of \cite{PPV}.
In section III.G of \cite{PPV} it is shown how to derive
a Fano-type converse from their Theorem 27.
The derivation and result generalise perfectly to
the entanglement-assisted codes for quantum channels:
As usual, the binary entropy is $h(p) := - (1-p)\log (1-p) - p \log p$,
and the binary relative entropy is
\begin{align}
        d(p\| q)
        :=& D((p,1-p)\|(q,1-q))\\
        =& p \log \frac{p}{q} + (1-p) \log \frac{1-p}{1-q}\\
        \geq&p \log \frac{1}{q} - h(p).\label{d-bound}
\end{align}
    
By the data processing inequality
    for quantum relative entropy under CPTP maps,
    and (\ref{d-bound})
    \begin{align}
        D(\rho_0 \| \rho_1)
        \geq&
        d(1-\epsilon\| \beta_{\epsilon}(\rho_0, \rho_1) )\\
        \geq&
        (1-\epsilon) \log \frac{1}{\beta_{\epsilon}(\rho_0,\rho_1)} 
        -  h(\epsilon)\label{beta-bound}.
    \end{align}
    Therefore,
    \begin{equation}
	D^{\all}_{\epsilon}(\rho_0 \| \rho_1 )
	\leq (D(\rho_0 \| \rho_1 ) + h(\epsilon))/(1-\epsilon) .
    \end{equation}
    Setting $\rho_0 = \chS \purho$ and
    $\rho_1 = \icS\oS$, and minimizing over $\oS$ yields
\begin{lemma}\label{fano}
    \begin{equation}
            \rIHch{\all}{\epsilon}{\chS}{\iS}
            \leq
            (I(\mc{E}_{\outS\qcon \A}, \rho_{\A} )
            + h(\epsilon))/
            (1-\epsilon)
    \end{equation}
\end{lemma}

The converse part of this theorem is easily derived from the
previous lemma, which tells us that
\begin{equation}
\begin{split}
    &\log \Mei{\epsilon}{(\chan^{\ox n})_{\outS^n \qcon \A^n}}
    {\rho_{\A^n}}\\
    \leq&
    \max_{\rho_{\A^n}}
    \frac{I((\chan^{\ox n})_{\outS^n \qcon \A^n},
    \rho_{\A^n} ) + h(\epsilon)}
    {1-\epsilon}.
\end{split}
\end{equation}
In \cite{CerfAdami} Adami and Cerf show that 
\begin{equation}
    I(\mc{E}^{(1)}_{\B_1\qcon\A_1}
\mc{E}^{(2)}_{\B_2\qcon\A_2}, \rho_{\A_1 \A_2})
    = I(\mc{E}^{(1)}_{\B_1\qcon\A_1},\rho_{\A_1})
+ I(\mc{E}^{(2)}_{\B_2\qcon\A_2},\rho_{\A_2}),
\end{equation}
so we have
\begin{equation}
	\max_{\rho_{\A^n}} I((\mc{E}^{\ox n})_{\B^n \qcon \A^n}, \rho_{\A^n})
\leq n \max_{\rho_{\A}} I(\mc{E}_{\B\qcon \A}, \iS),
\end{equation}
and
\begin{equation}
    \lim_{n\to\infty}
    \frac{1}{n} \log
    M_{\epsilon}((\mc{E}^{\ox n})_{\B^n \qcon \A^n})
    \leq
    \frac{1}{1-\epsilon}
    \max_{\iS} I(\chS,\iS).
\end{equation}
Taking the limit $\epsilon \to 0$ 
completes the proof.

\subsection{Using symmetries}\label{sym}
We now show how symmetries of $\chS$ can be used to simplify
the computation of $I^{\goc}_{\epsilon}(\chS,\iS)$.
For classical channels, analogous results were
obtained in \cite{2013-Polyanskiy} (but note that
\cite{2013-Polyanskiy} also deals
with infinite input/output alphabets)
and similar ideas were discussed in \cite{2012-Matthews}.

Suppose that there is a group $G$ with a
representation
$g \mapsto g_{\A|\A} \in \ops(\A \to \A)$
for all $g \in G$, given by
\begin{equation}
    g_{\A|\A} \tau_{\A} = U(g)_{\A} \tau_{\A} U(g)_{\A}\hc,
\end{equation}
and a representation
$g \mapsto g_{\B|\B} \in \ops(\B \to \B)$, given by
\begin{equation}
     g_{\B|\B} \tau_{\B} = V(g)_{\B} \tau_{\B} V(g)_{\B}\hc,
\end{equation}
where $U$ and $V$ are unitary representations of $G$,
and let
\begin{equation}
    g_{\At|\At} \tau_{\At} := U(g)^{\ast}_{\A} \tau_{\At} U(g)_{\A}^{\T}.
\end{equation}

Suppose that $\chS$ possesses the \emph{$G$-covariance}
\begin{equation}
    \forall g \in G:~\chS g_{\A|\A}
    = 
    g_{\B|\B} \chS.
\end{equation}
\begin{proposition}\label{sym1} 
    For any $\Omega \supseteq \LO$ and for all $g \in G$,
    \begin{equation}
        \begin{split}
        &\rDH{\goc}{\epsilon}
        {\chS g_{\At|\At} g_{\A|\A} \purho}
        {g_{\At|\At}g_{\B|\B} \rho_{\At} \oS}\\
        =&
        D^{\goc}_{\epsilon}(\chS \purho
        \|
        \rho_{\At} \sigma_{\outS})
        \end{split}
    \end{equation}
    and
    $I^{\goc}_{\epsilon}(\chS, g_{\A|\A} \iS)
    = I^{\goc}_{\epsilon}(\chS, \iS)$.
\end{proposition}
\begin{proof}
    The first claim follows from
    \begin{align}
    &\chS g_{\A|\A} \purho\\
	=&
    \chan_{\outS\qcon\A}[
    U(g)_{\A}\rho^{\half}_{\A}U(g)_{\A}^{\dagger}
    \Phi_{\At\A}
    U(g)_{\A}\rho^{\half}_{\A}U(g)_{\A}^{\dagger}
    ]\\
    =&
    \chan_{\outS\qcon\A}[
    U(g)_{\A}\rho^{\half}_{\A}U(g)_{\At}^{\ast}
    \Phi_{\At\A}
    U(g)^{\T}_{\At}\rho^{\half}_{\A}U(g)_{\A}^{\dagger}
    ]\\
    =&
    V(g)_{\outS}U(g)_{\At}^{\ast}
    \chan_{\outS\qcon\A}[
    \rho^{\half}_{\A}
    \Phi_{\At\A}
    \rho^{\half}_{\A}
    ]
    V(g)_{\outS}U(g)_{\At}^{\T}\\
    =& g_{\At|\At} g_{\B|\B} \chS \purho,
\end{align}
the fact that $g_{\At|\At} g_{\B|\B}$ and
its inverse belong to
$\LO(\At\to\At,\outS\to\outS) \subseteq \goc$,
and Corollary \ref{revdpi}. We use this to prove the second claim thus:
\begin{align}
    &I^{\goc}_{\epsilon}(\chS, \iS)\\
    =& \min_{\sigma_{\outS}}
    D^{\goc}_{\epsilon}(\chS \purho
    \|
    \rho_{\At} \sigma_{\outS})\\
    =& D^{\goc}_{\epsilon}(\chS \purho
    \|
    \rho_{\At} \sigma^{0}_{\outS})\\
    =& D^{\goc}_{\epsilon}(
    \chS g_{\At|\At} g_{\A|\A} \purho 
    \|
    g_{\At|\At}g_{\B|\B} \rho_{\At}
     \sigma^{0}_{\outS})\\
    \geq & \min_{\sigma_{\outS}}
    D^{\goc}_{\epsilon}
    (\chS g_{\At|\At} g_{\A|\A} \purho 
    \|
    g_{\At|\At} \rho_{\At} \sigma_{\outS})\\
    =& I^{\goc}_{\epsilon}(\chS, g_{\A|\A} \iS).
\end{align}
Since $g$ has an inverse in $G$, the reverse inequality holds too.
\end{proof}
Suppose that there is a Haar ($G$-invariant) measure $\mu$ on $G$, and let
$\bar{\rho}_{\A} := \int_{G} d\mu(g) g_{\A|\A} \iS $. Then, $\bar{\rho}_{\A}$
is invariant under the action $g_{\A|\A}$, and by Jensen's inequality,
Corollary \ref{convexity},
and Proposition \ref{sym1},
\begin{align}
\rBch{\mclass}{\epsilon}{\chS}{\bar{\rho}_{\A}}
    \leq&
\int_{G} d\mu(g) \rBch{\mclass}{\epsilon}{\chS}{g_{\A|\A}\rho_{\A}}\\
    =& \rBch{\mclass}{\epsilon}{\chS}{\iS}.
\end{align}
Therefore, the optimisation over $\iS$ can be restricted to
those density operators invariant under the action of $g_{\A|\A}$.

An important type of symmetry that an operation
representing $n$ channel uses may possess is
permutation covariance. For example, this applies
to $n$ uses of a memoryless channel.

For any element $\pi$ of the symmetric group $S_n$,
and $n$-partite system $\Q^n := \Q_1 \Q_2 \ldots \Q_n$ consisting
of $n$ isomorphic systems $\Q_j$, let
$\pi_{\Q^n|\Q^n} \in \ops(\Q^n \to \Q^n)$ denote the
unitary operation which permutes the $n$ systems.

An operation
$\chan_{\outS^n \qcon \A^n} \in \ops(\A^n \to \outS^n)$
is permutation covariant if
\begin{equation}
	\chan_{\outS^n \qcon \A^n} \pi_{\A^n|\A^n}
	= \pi_{\outS^n|\outS^n} \chan_{\outS^n \qcon \A^n}.
\end{equation}
\def\swap{\chi}
Suppose that,
in addition to permutation invariance of the $n$ uses,
each \emph{use} of the channel is $G$-covariant
in the sense that
\begin{equation}
	\chan_{\outS^n \qcon \A^n} g_{\A_j|\A_j}
	= g_{\B_j|\B_j} \chan_{\outS^n \qcon \A^n}
\end{equation}
for all $g \in G$ and $j \in \{1, \ldots, n\}$.
Here the representations of $G$ on each system are the same, 
except that they act on different systems.
This is the case, for example, if $\chan_{\outS^n \qcon \A^n}$
is $n$ uses of a $G$-covariant memoryless channel.

To every ordered pair $(\pi,\mathbf{g})$ where $\pi \in S_n$
and $\mathbf{g} \in G^{\times n}$, we can associate an
action $(\pi,\mathbf{g})_{\A^n|\A^n}
:= \pi_{\A^n|\A^n} \mathbf{g}_{\A^n|\A^n}$.
Here the action on $\A^n$ for $\mathbf{g} = (g^{(1)}, \ldots, g^{(n)})$
is $\mathbf{g}_{\A^n|\A^n} := \Ox_{j=1}^n g_{\A_j|\A_j}^{(j)}$
and $\mathbf{g}_{\B^n|\B^n} := \Ox_{j=1}^n g_{\B_j|\B_j}^{(j)}$.
Under composition
these actions constitute a group, which is a semi-direct
product of $G^{\times n}$ and $S_n$ ($G^{\times n}$ being the
normal subgroup) which we denote $S_n \ltimes G^{\times n}$.
Defining the action of $(\pi, \mathbf{g}) \in S_n \ltimes G^{\times n}$
on states of $\outS^n$ by
$
	(\pi, \mathbf{g})_{\outS^n|\outS^n}
	:= \pi_{\outS^n|\outS^n} \mathbf{g}_{\outS^n|\outS^n},
$
we have
\begin{equation}
	\chan_{\outS^n \qcon \A^n} (\pi, \mathbf{g})_{\A^n|\A^n}
	= (\pi, \mathbf{g})_{\outS^n|\outS^n} \chan_{\outS^n \qcon \A^n}.
\end{equation}

\section{Example: The depolarising channel}\label{sec:eg}
\def\dpc{\Delta}

	A single use of the $d$-dimensional depolarising
	channel with parameter $p$ and $d$-dimensional input
	and output systems $\A$ and $\outS$ has the operation
	\begin{equation}
		\mathcal{D}_{\outS \qcon \A} \tau_{\A} 
		= (1-p)\tau_{\outS} + p \tr(\tau_{\outS}) \mu_{\outS},
	\end{equation}
	where, for any system $\Q$, $\mu_{\Q} := \1_{\Q}/\dim(\Q)$
	denotes the maximally mixed state on that system.
	For $n$ uses the operation is
	\begin{equation}
		\mathcal{D}^{\ox n}_{\outS^n \qcon \A^n} =
		\mathcal{D}_{\outS_1 \qcon \A_1}
		\ldots \mathcal{D}_{\outS_n \qcon \A_n},
	\end{equation}
	which has the covariance group $S_n \ltimes \mathrm{U}(d)^{\times n}$.

    The only input and output states with the corresponding invariances
    are the maximally mixed states. Therefore,
    \begin{align}
        \sup_{\rho_{\A^n}} I^{\all}_{\epsilon}
        (\mathcal{D}^{\ox n}_{\outS^n \qcon \A^n} ,
        \rho_{\A^n})
        =
        D^{\all}_{\epsilon} (
        \phi(p)^{\ox n}_{\At^n \outS^n}
        \|
        \mu_{\At^n} \mu_{\outS^n} )
    \end{align}
    where
    $
    \phi(p)_{\At \outS} :=
    \mathcal{D}_{\outS \qcon \A}[\Phi_{\At\A}/d]
    =
        (1-p) \Phi_{\At \outS}/d + p \mu_{\At} \mu_{\outS}    
    $
	is an \emph{isotropic} state.
    Since the arguments of $D^{\all}_{\epsilon}$ commute in this expression,
    this is equivalent
    to a classical hypothesis test between the distributions given
    by the eigenvalues of the two states.
    In fact, the degeneracy of the eigenvalues makes it is equivalent
    to deciding between hypotheses on the distribution of
    $n$ samples of a binary variable:
    Hypothesis $H_0$ is that the samples
    are drawn i.i.d. with probability $(1-p) + p/d^2$ of being $0$,
    and hypothesis
    $H_1$ is that the samples are drawn i.i.d. with probability $1/d^2$
    of being $0$. Therefore,
    \begin{equation}
        \begin{split}
        &D^{\all}_{\epsilon} (
        \phi(p)^{\ox n}_{\At^n \outS^n}
        \|
        \mu_{\At^n} \mu_{\outS^n} )\\
        =&
        D^{\all}_{\epsilon}( (\mu,1-\mu)^{\ox n},
        (\lambda, 1-\lambda)^{\ox n} )
        \end{split}
    \end{equation}
    where $\mu = (1-p) + p/d^2$ and $\lambda = 1/d^2$,
    and Proposition \ref{bernoulli-beta} gives a formula for this quantity
    which is easy to evaluate exactly (as we have done for
     Fig. \ref{fig:depol-graph}).

\begin{proposition}\label{bernoulli-beta}
    Let $\mu \geq \lambda$ be two probabilities.
    \begin{equation}
        \beta_{\epsilon}( (\mu,1-\mu)^{\ox n}, (\lambda, 1-\lambda)^{\ox n} )
        = (1-\gamma) \beta_{\ell(\epsilon)}
        + \gamma \beta_{\ell(\epsilon)+1}
    \end{equation}
    where
    \begin{align}
        \alpha_{\ell}
        =& \sum_{j=0}^{\ell-1} \binom{n}{j} \mu^{j}(1-\mu)^{n-j}, \\
        \beta_{\ell}
        =& \sum_{j=l}^{n} \binom{n}{j} \lambda^{j}(1-\lambda)^{n-j},
    \end{align}
    $\ell(\epsilon)$ is the value of $l$ satisfying
    $\alpha_{\ell} \leq \epsilon \leq \alpha_{\ell + 1}$
    and
    $\gamma = (\alpha - \alpha_{\ell(\alpha)})/
                (\alpha_{\ell(\alpha)+1} - \alpha_{\ell(\alpha)})$.
\end{proposition}
\begin{proof}
    This is just optimising over the optimal (classical) hypothesis tests
    identified by the Neyman-Pearson lemma. The same expression is
    given in \cite{PPV}.
\end{proof}

\begin{figure}[h]
\includegraphics[scale=1.0]{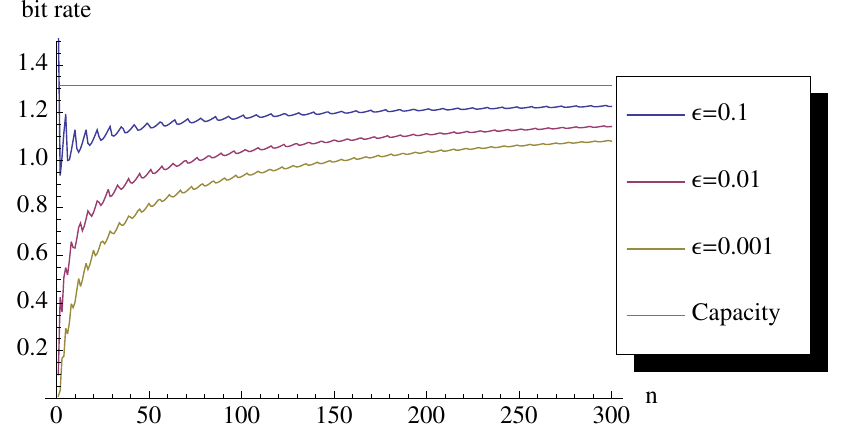}
\caption{
    Our upper bound (Theorem \ref{se-converse}) on the rate
    of entanglement-assisted codes
    evaluated (see Section \ref{sec:eg})
    for three different error probabilities $\epsilon$,
    for the qubit depolarising channel with failure probability
    $0.15$.
	The red line marks the capacity of the channel (roughly
	1.31 bits/channel use)
	as given by the formula of Bennett, Shor, Smolin
	and Thapliyal~\cite{BSST}.
}
\label{fig:depol-graph}
\end{figure}

\section{Converses and cryptography}\label{sec:sconv}
In~\cite{kw:converse} a strong converse was proven for the classical capacity
of many quantum channels $\mathcal{N}$, including depolarising noise. 
The results of~\cite{kw:converse} had a nice application to
proving security in so-called noisy-storage model of 
quantum cryptography~\cite{noisy:new}. This model
allows for the secure implementation of any two-party cryptographic
task under the assumption that the adversary's quantum memory is noisy.
Examples of such tasks include bit commitment and oblivious 
transfer, which are impossible to achieve without
assumptions~\cite{lo,mayers}. 

Concretely, the noisy-storage model assumes that during waiting times introduced
into the protocol, the adversary can only store quantum information in a memory device modelled
by a channel $\chS$. Otherwise the adversary is all powerful and may even use a quantum computer
to peform the most advantageous error-correcting encoding. At the beginning of the waiting time, the state
of the protocol can be described as a cq-state
$\rho_{\mathrm{X}\mathrm{K}\A}$
where $\mathrm{X}$ is a classical register held by the honest party
and storing a string $X$
and $\A$ and $\mathrm{K}$
are the quantum and classical registers of the
dishonest party respectively. 
After the wait time, the state of the protocol is described by 
$\chS \rho_{\mathrm{X}\mathrm{K}\A}$.
The security of all protocols proposed in this model requires a bound on the min-entropy ${\rm H}_{\rm min}(\mathrm{X}|\B\mathrm{K}) = - \log P_{\rm guess}(\mathrm{X}|\B\mathrm{K})$, 
where $P_{\rm guess}(\mathrm{X}|\B\mathrm{K})$ is the probability that the adversary holding $\B$ and $\mathrm{K}$
manages to guess $X$ maximized over all possible measurements
on $\B\mathrm{K}$.  
Bounds on this quantity can be linked to the classical capacity~\cite{noisy:new}, the entanglement cost~\cite{entCost} 
or the quantum capacity~\cite{bfw12} of $\chS$. Yet, explicit bounds on the min-entropy which would enable practical implementations
of such protocols~\cite{bcExp,otExp} are elusive. In particular, up to this date no statements are known for arbitrary channels $\chS$, such as channels which do not obey a strong converse or simply structureless channels. Our method can be used to compute explicit bounds on the min-entropy for arbitrary channels.

The key to relating our analysis to a study of the min-entropy is the observation~\cite{noisy:new} that
\begin{align}
&{\rm H}_{\rm min}(\mathrm{X}|\B\mathrm{K}) \geq\\
&\qquad - \log \max_{\mathcal{Z}} \Pr\left(\hat{W} = W|\chS,\mathcal{Z}, S_{\lfloor{\rm H}_{\rm min}(\mathrm{X}|\mathrm{K})\rfloor}\right)\ .
\end{align}
That is, by understanding the adversary's knowledge
${\rm H}_{\rm min}(\mathrm{X}|\mathrm{K})$ conditioned on the classical knowledge $\mathrm{K}$ alone, and the properties of
the channel $\chS$ we can bound the adversary's knowledge about a
string $\mathrm{X}$ given \emph{both} $\B$ and $\mathrm{K}$. 
In~\cite{noisy:new} the bound ${\rm H}_{\rm min}(\mathrm{X}|\mathrm{K}) \gtrsim n/2$ 
was obtained from an uncertainty relation for BB84 measurements that
was used to generate the $n$-bit string $X$.
Different measurements lead to higher (or lower)
values of 
\begin{align}
\frac{{\rm H}_{\rm min}(\mathrm{X}|\mathrm{K})}{n} =: \hat{R}\ .
\end{align}
In~\cite{noisy:new}, a strong converse for some channels
$\mathcal{N}^{\otimes \ell}$ was then used to bound the r.h.s.
for $\chS$ and rate $R = (n\hat{R})/\ell$, as long as $R$ exceeded the capacity of $\mathcal{N}$.

Let us now sketch how 
our approach directly leads to a security statement for any $\chS$. 
More specifically, we will turn things around, fix $\hat{R}$, and ask how large we
have to choose $n$ such that sending $n\hat{R}$ randomly chosen classical bits
through the channel $\chS$ will incur an error of at least
$\epsilon$. Intuitively, our goal is to effectively overflow
the adversary's storage device $\chS$, that is, $n$ will be chosen such that no coding scheme allows for an error less than $\epsilon$ for 
a fixed value of $\hat{R}$. By the results of~\cite{noisy:new} 
we can then bound the adversary's min-entropy as
${\rm H}_{\rm min}(\mathrm{X}|\B\mathrm{K}) \geq - \log (1-\epsilon)$. 
For any $\epsilon > 0$, our analysis yields a bound on
$\log M_{\epsilon}$ stating that if we were to transmit more
than $n \hat{R} > \log M_{\epsilon}$ bits, the error is
necessarily $\epsilon + \delta > \epsilon$ for some $\delta > 0$, 
which yields the desired bound. As $\hat{R}$ is determined by an
uncertainty relation, we thus know how many transmissions $n$ 
we have to make to obtain security.
Note that allowing the adversary to perform entanglement-assisted coding only gives him additional power, and hence our
analysis for entanglement-assisted coding provides a method with which explicit security parameters may be computed by
means of a semi-definite program.

\section{Conclusion}\label{sec:conc}
We have shown how a simple and powerful
idea \cite{PPV} for obtaining a finite blocklength
converse for classical channels in terms
of a hypothesis testing problem 
can be generalised to quantum channels
and entanglement-assisted codes.

This generalisation has the property that
a natural restriction on codes
(removing entanglement-assistance)
translates into a natural restriction
on the tests that can be performed
in the hypothesis testing problem
(they must be local). This provides
a strong link between the extensively
studied problems of channel coding and
hypothesis testing (and state discrimination)
of bipartite systems under locality
restrictions.

Many avenues for further work are apparent
to the authors:
Subsection \ref{sym} invites a more thorough investigation
into the extent to which symmetries can be
used to simplify evaluation of the bound
in various special cases, and
especially in the case of general memoryless
channels, where one might hope for an
exponential reduction in the size of
the SDP, in a quantum generalisation
of \cite{2012-Matthews}.

While we have been able to fit the existing
converse of Wang and Renner \cite{RennerWang} precisely into
our hierarchy bounds based on restricted
hypothesis testing, it is not obvious to
the authors what relationship exists between
our converse for entanglement-assisted codes
and that of Datta and Hsieh \cite{DH}. We would
like to know what can be said about this,
particularly in light of the achievability
bound given in \cite{DH}.

A limitation of our work is that we only
generalise the classical bound of \cite{PPV}
for the case of finite input/output alphabets
(as our input/output systems have finite
dimension). To analyse the general case
will require greater mathematical sophistication
(see \cite{2013-Polyanskiy}), but would be desirable given
(for example) the interest in quantum gaussian channels \cite{Gaussian}.

\acknowledgments
Matthews acknowledges the support of the Isaac Newton Trust,
NSERC and QuantumWorks
and would like to thank Andreas Winter, Debbie Leung, and Nilanjana Datta for
useful conversations regarding this work.
SW was supported by the National Research Foundation
and the Ministry of Education, Singapore.
We would like to thank the anonymous referees for
their careful reviews of our first draft.

\newpage

\bibliographystyle{unsrt} 
\bibliography{bigbib}

\begin{thebibliography}{10}

\bibitem{MD}
M.~Milan and N.~Datta.
\newblock Generalized relative entropies and the capacity of classical-quantum
  channels.
\newblock {\em Journal of Mathematical Physics}, 50(7):072104, 2009.

\bibitem{RennerWang}
L.~Wang and R.~Renner.
\newblock One-shot classical-quantum capacity and hypothesis testing.
\newblock {\em Phys. Rev. Lett.}, 108:200501, May 2012.

\bibitem{RR}
J.M. Renes and R.~Renner.
\newblock Noisy channel coding via privacy amplification and information
  reconciliation.
\newblock {\em Information Theory, IEEE Transactions on}, 57(11):7377--7385,
  2011.

\bibitem{DH}
N.~{Datta} and M.-H. {Hsieh}.
\newblock {One-shot entanglement-assisted quantum and classical communication}.
\newblock {\em ArXiv e-prints}, May 2011.

\bibitem{DMHB}
N.~Datta, M.~Mosonyi, M-H. Hsieh, and F.~G. S.~L. Brandao.
\newblock Strong converse capacities of quantum channels for classical
  information.
\newblock {\em To appear in IEEE Trans. Inf. Th}, 2011.

\bibitem{PPV}
Y.~Polyanskiy, H.~V. Poor, and S.~Verd\'u.
\newblock Channel coding rate in the finite blocklength regime.
\newblock {\em IEEE Transactions on Information Theory}, pages 2307--2359,
  2010.

\bibitem{2013-Polyanskiy}
Y.~Polyanskiy.
\newblock Saddle point in the minimax converse for channel coding.
\newblock {\em Information Theory, IEEE Transactions on}, 59(5):2576--2595,
  2013.

\bibitem{Rains-rig}
E.~M. Rains.
\newblock Rigorous treatment of distillable entanglement.
\newblock {\em Phys. Rev. A}, 60:173--178, Jul 1999.

\bibitem{2001-Rains-SDP}
E.M. Rains.
\newblock A semidefinite program for distillable entanglement.
\newblock {\em Information Theory, IEEE Transactions on}, 47(7):2921 --2933,
  nov 2001.

\bibitem{VP}
S.~Virmani and M.~B. Plenio.
\newblock Construction of extremal local positive-operator-valued measures
  under symmetry.
\newblock {\em Phys. Rev. A}, 67:062308, Jun 2003.

\bibitem{2005-DArianoLoPrestiPerinotti}


\bibitem{BSST}
C.~H. Bennett, P.~W. Shor, J.~A. Smolin, and A.~V. Thapliyal.
\newblock Entanglement-assisted capacity of a quantum channel and the reverse
  {S}hannon theorem.
\newblock {\em Information Theory, IEEE Transactions on}, 48(10):2637--2655,
  October 2002.

\bibitem{2012-Matthews}
W.~Matthews.
\newblock A linear program for the finite block length converse of
  {P}olyanskiy, {P}oor, {V}erd\'{u} via nonsignaling codes.
\newblock {\em Information Theory, IEEE Transactions on}, 58(12):7036--7044,
  2012.

\bibitem{2010-PolyanskiyVerdu}
Y.~Polyanskiy and S.~Verdú.
\newblock Arimoto channel coding converse and {R}\'{e}nyi divergence.
\newblock In {\em Communication, Control, and Computing (Allerton), 2010 48th
  Annual Allerton Conference on}, pages 1327--1333, Sept 2010.

\bibitem{2012-SharmaWarsi}
Naresh Sharma and Naqueeb~Ahmad Warsi.
\newblock Fundamental bound on the reliability of quantum information
  transmission.
\newblock {\em Phys. Rev. Lett.}, 110:080501, Feb 2013.

\bibitem{2006-HayashiBook}
Masahito Hayashi.
\newblock {\em Quantum Information: An Information}.
\newblock Springer, 2006.

\bibitem{VB-SDP}
L.~Vandenberghe and S.~Boyd.
\newblock Semidefinite programming.
\newblock {\em SIAM review}, 38(1):49--95, 1996.

\bibitem{holevo1973}
A.S. Holevo.
\newblock Bounds for the quantity of information transmitted by a quantum
  communication channel.
\newblock {\em Problemy Peredachi Informatsii}, 9(3):3--11, 1973.

\bibitem{2001-SchumacherWestmoreland}
Benjamin Schumacher and Michael~D. Westmoreland.
\newblock Optimal signal ensembles.
\newblock {\em Phys. Rev. A}, 63:022308, Jan 2001.

\bibitem{HSW-H}
A.S. Holevo.
\newblock The capacity of the quantum channel with general signal states.
\newblock {\em Information Theory, IEEE Transactions on}, 44(1):269 --273, jan
  1998.

\bibitem{HSW-SW}
B.~Schumacher and M.~D. Westmoreland.
\newblock Sending classical information via noisy quantum channels.
\newblock {\em Phys. Rev. A}, 56:131--138, Jul 1997.

\bibitem{CerfAdami}
C.~Adami and N.~J. Cerf.
\newblock von neumann capacity of noisy quantum channels.
\newblock {\em Phys. Rev. A}, 56:3470--3483, Nov 1997.

\bibitem{kw:converse}
R.~K{\"o}nig and S.~Wehner.
\newblock A strong converse for classical channel coding using entangled
  inputs.
\newblock {\em Physical Review Letters}, 103:070504, 2009.

\bibitem{noisy:new}
R.~K\"onig, S.~Wehner, and J.~Wullschleger.
\newblock Unconditional security from noisy quantum storage.
\newblock {\em IEEE Transactions on Information Theory}, 58(3):1962 -- 1984,
  2012.

\bibitem{lo}
H-K. Lo and H.F. Chau.
\newblock Why quantum bit commitment and ideal quantum coin tossing are
  impossible.
\newblock {\em Physica D: Nonlinear Phenomena}, 120(1–2):177 -- 187, 1998.
\newblock Proceedings of the Fourth Workshop on Physics and Consumption.

\bibitem{mayers}
D.~Mayers.
\newblock Unconditionally secure quantum bit commitment is impossible.
\newblock {\em Phys. Rev. Lett.}, 78:3414--3417, Apr 1997.

\bibitem{entCost}
M.~Berta, M.~Christandl, F.G.S.L. Brandao, and S.~Wehner.
\newblock Entanglement cost of quantum channels.
\newblock In {\em Information Theory Proceedings (ISIT), 2012 IEEE
  International Symposium on}, pages 900 --904, july 2012.

\bibitem{bfw12}
M.~Berta, O.~Fawzi, and S.~Wehner.
\newblock Quantum to classical randomness extractors.
\newblock In Reihaneh Safavi-Naini and Ran Canetti, editors, {\em Advances in
  Cryptology – CRYPTO 2012}, volume 7417 of {\em Lecture Notes in Computer
  Science}, pages 776--793. Springer Berlin / Heidelberg, 2012.

\bibitem{bcExp}
N.~Ng, S.~Joshi, C.~Chia, C.~Kurtsiefer, and S.~Wehner.
\newblock Experimental implementation of bit commitment in the noisy-storage
  model.
\newblock {\em Nature Communications}, 3(1326), 2012.

\bibitem{otExp}
C.~Erven, N.~Ng, N.~Gigov, R.~Laflamme, S.~Wehner, and G.~Weihs.
\newblock An experimental implementation of oblivious transfer in the noisy
  storage model.
\newblock {\em Nature Communications}, 5:3418, 2014.

\bibitem{Gaussian}
A.~S. Holevo, M.~Sohma, and O.~Hirota.
\newblock Capacity of quantum gaussian channels.
\newblock {\em Phys. Rev. A}, 59:1820--1828, Mar 1999.

\end{thebibliography}

\end{document}